%% file: main.tex
\documentclass[11pt,letterpaper]{article}
\nonstopmode
\usepackage[margin=1in]{geometry}

\usepackage{amsmath,amssymb,amsthm}
\usepackage{tikz}
\usetikzlibrary{arrows.meta,calc}
\usepackage{caption}
\usepackage{placeins,booktabs,tabularx}

\usepackage{graphicx}
\usepackage{microtype}
\usepackage{algorithm}
\usepackage{algorithmicx}
\usepackage{algpseudocode}
\usepackage{braket}
\usepackage{quantikz}
\definecolor{circuitInk}{HTML}{18334A}
\definecolor{circuitSource}{HTML}{25618D}
\definecolor{circuitAux}{HTML}{7151A0}
\definecolor{circuitClean}{HTML}{14786E}
\tikzset{
  circuit wire/.style={draw=circuitInk,line width=.65pt},
  circuit gate/.style={draw=circuitSource,fill=circuitSource!6,
    rounded corners=2pt,line width=.85pt,align=center,font=\normalsize},
  circuit auxiliary/.style={circuit gate,draw=circuitAux,fill=circuitAux!6},
  circuit tag/.style={draw=circuitClean,fill=circuitClean!5,
    rounded corners=2pt,inner xsep=5pt,inner ysep=4pt,
    text=circuitClean,font=\small},
  circuit label/.style={font=\small,text=circuitInk}
}
\usepackage{hyperref}
\usepackage{cleveref}
\hypersetup{colorlinks=true,linkcolor=circuitInk,citecolor=circuitInk,urlcolor=circuitSource,
  pdftitle={Unitary RQL Equals RQL},pdfauthor={Quinten Tupker},
  pdfsubject={Technical manuscript accompanying the QIP extended abstract}}

\newtheorem{theorem}{Theorem}[section]
\newtheorem{lemma}[theorem]{Lemma}
\newtheorem{proposition}[theorem]{Proposition}
\newtheorem{corollary}[theorem]{Corollary}
\newtheorem{definition}[theorem]{Definition}
\newtheorem{remark}[theorem]{Remark}

\newcommand{\RL}{\mathsf{RL}}

\newcommand{\DSPACE}{\mathsf{DSPACE}}

\newcommand{\BQL}{\mathsf{BQL}}
\newcommand{\BQUL}{\mathsf{BQ_UL}}

\newcommand{\RQL}{\mathsf{RQL}}
\newcommand{\RQUL}{\mathsf{RQ_UL}}
\newcommand{\QTM}{\mathsf{QTM}}
\newcommand{\UQTM}{\mathsf{UQTM}}

\newcommand{\Lclass}{\mathsf{L}}
\newcommand{\BPL}{\mathsf{BPL}}
\newcommand{\PL}{\mathsf{PL}}
\newcommand{\QL}{\mathsf{QL}}

\DeclareMathOperator{\Tr}{Tr}
\DeclareMathOperator{\poly}{poly}

\newcommand{\norm}[1]{\left\lVert #1 \right\rVert}

\providecommand{\proj}[1]{\ket{#1}\!\bra{#1}} 
\DeclareMathOperator{\vecf}{vec}

\title{Unitary $\RQL$ Equals $\RQL$}
\author{
  Quinten Tupker\thanks{CWI, Amsterdam, The Netherlands. Email: qmtupker@gmail.com}
}
\date{}

\input{presentation-options}
\begin{document}

\maketitle

\begin{abstract}
Intermediate measurements let quantum computations discard information
and reuse their workspace. Deferring all measurements can require storing
their entire history, which need not preserve logarithmic space.
Fefferman and Remscrim proved that measurements can nevertheless be
eliminated with two-sided bounded error, and asked whether the same holds
with one-sided error~\cite{FeffermanRemscrim21}.

We prove $\RQL_\Gamma=\RQUL_\Gamma$ for the standard gate set
$\Gamma=\{H,T,\mathrm{CNOT}\}$, preserving polynomial time, logarithmic
space, and exactly zero acceptance on no-instances. Our proof extends
the density-matrix doubling method of Girish, Raz and Zhan
\cite{GirishRazZhan21} to general channels, including resets and classical
memory. We also prove measurement elimination for broader finite gate
sets with exact inverses, including suitable gates with transcendental
entries, and establish gate-set independence for a specified family of
algebraic gate sets.

A separate history-checking construction proves
$\mathsf{QMAL}_{1,\mathcal G}=\mathsf{Q_UMAL}_{1,\mathcal G}$
under explicit gate assumptions: measurements can also be eliminated
from logarithmic-space quantum verification while preserving perfect
completeness.
\end{abstract}

\input{intro}
\input{prelims}
\input{result}
\input{result-qmal}
\input{conclusion}

\bibliographystyle{alpha}
\bibliography{bibliography,submission-references}

\clearpage
\appendix
\input{appendix-doubling}
\input{result-bql}
\input{appendix-background}
\input{appendix-formal}

\end{document}

%% file: presentation-options.tex
\newif\ifStandardGateLead
\StandardGateLeadtrue

%% file: intro.tex
\section{Introduction}
Measurements allow a quantum computation to extract classical information,
discard information, and reuse its workspace. In polynomial time, these
operations can be simulated by a unitary circuit that retains the discarded
information in ancillary qubits and measures only at the end
\cite{AharonovKitaevNisan98,NielsenChuang10}. This construction need not
preserve a space bound: recording $T$ single-qubit measurement outcomes
can require $T$ additional qubits. We ask whether measurements can be
eliminated while retaining logarithmic workspace and a one-sided error
guarantee.

\subsection{Background}
\label{sec:intro-background}
Watrous established the basic complexity-theoretic framework for
space-bounded quantum computation and its classical simulation bounds
\cite{watrous1999space,Watrous03}. We consider polynomial-time computations
with $O(\log n)$ quantum work qubits and classical work bits, and read-only
access to an input of length $n$.

The class $\BQL$ allows intermediate measurements and classical control.
It requires acceptance probability at least $2/3$ on yes-instances and at
most $1/3$ on no-instances. Its unitary counterpart, $\BQUL$, allows only
a final measurement. Fefferman and Remscrim proved
$\BQL=\BQUL$, preserving logarithmic space and polynomial time
\cite[Theorem~1]{FeffermanRemscrim21}.

The one-sided class $\RQL$ instead requires acceptance probability exactly
zero on no-instances; $\RQUL$ is its unitary counterpart. Both models here
are absolutely halting: every branch on every input halts within a common
polynomial bound. An additive approximation preserves a two-sided gap but
can change zero acceptance into a small positive probability. Fefferman
and Remscrim therefore left $\RQL=\RQUL$ as a separate question
\cite[Open Problem~1]{FeffermanRemscrim21}.

Girish, Raz and Zhan independently eliminated measurements for unital
computations, whose channels preserve the maximally mixed state
\cite[Section~5]{GirishRazZhan21}. Their density-matrix doubling method
represents a mixed state by a vector on two registers. Unitality bounds
the norms of the resulting operators, but excludes a qubit reset. We
extend this construction to general channels while preserving zero
acceptance on no-instances.

Exact acceptance conditions also make the gate convention important.
Approximate gate synthesis need not preserve a zero probability
\cite{MelkebeekWatson12,LeGallLiuWang23}. We therefore state the gate
assumptions explicitly. Figure~\ref{fig:logspace-map} summarizes prior
logarithmic-space results under the cited gate conventions;
Appendix~\ref{app:space-background} gives further background.
\input{space-complexity-map}
\FloatBarrier

\input{intro-main-results}

\subsection{Proof idea: preserve zero acceptance and retain a signal}
Density-matrix doubling represents $X$ by
$\vecf(X)=\sum_{i,j}X_{ij}\ket i\ket j$.
For a channel $\Phi_t(X)=\sum_aK_{t,a}XK_{t,a}^{\dagger}$, its doubled
operator is $\sum_aK_{t,a}\otimes\overline{K_{t,a}}$.
The two registers use twice the original number of qubits.

The one-sided argument starts from a simple observation. If $K_h$ is
the product of Kraus operators along history $h$, and $\Pi$ projects onto
acceptance, then
\[
 p=\sum_h\|\Pi K_h\ket0\|^2=0
 \quad\Longrightarrow\quad \Pi K_h\ket0=0\quad\text{for every }h.
\]
Consequently, changing only the right factor in each doubled term cannot
create acceptance: every accepting term still contains the same zero
left vector. Accuracy is needed to retain acceptance on yes-instances.

The difficulty is normalization. A doubled qubit reset has operator
norm $\sqrt2$, so normalizing every step separately can lose an
exponential factor. We use reference states that follow the source
channels, with a small maximally mixed contribution added at each step.
Their square roots and inverse square roots reweight the right factors.
Each reweighted step has norm at most one, consecutive weights cancel,
and the final accepting probability remains inverse polynomial.
These weights are related to Petz recovery
\cite{Petz86,Petz88,BarnumKnill02,GilyenLloydMarvianQuekWilde22};
Section~\ref{sec:doubling-framework} explains the relation precisely.

A unitary circuit implements an approximation to each reweighted step
in the part selected by zero scratch registers, while retaining the
left Kraus factors exactly. A logarithmic counter selects the product
of those parts without storing all outcomes. Unitary amplification then
raises the accepting probability to a constant.
Figure~\ref{fig:rql-roadmap} summarizes the construction.
\input{figure-rql-roadmap}
\FloatBarrier

Section~\ref{sec:rql-standard} proves the Clifford+$T$ case by following
this algorithm. Sections~\ref{sec:rql-adjoints} and
\ref{sec:rql-algebraic} then address broader gate sets.
The implementation of the right factors uses the established two-sided
measurement-elimination theorem: reference-entry estimates are converted
into unitary circuits before the one-sided simulation is constructed.
There is no dependence on the one-sided conclusion. The bounded-error
application in Appendix~\ref{sec:doubling-bql} uses the same implementation
and is not an independent proof of $\BQL=\BQUL$.

Section~\ref{sec:qmal-one} gives a separate construction for quantum
witnesses with perfect completeness. It tests a history of rejection
observables supplied by Merlin and uses neither reweighting nor two-sided
measurement elimination. Supporting mathematical lemmas are checked in
Lean; Appendix~\ref{app:formal-verification} states the scope.

\paragraph{Use of generative AI.}
AI tools were used in drafting this article. They were also used in
developing and checking parts of the mathematical arguments, searching
and checking the literature,
revising the exposition, and constructing formal proofs of
supporting lemmas.
The Lean proofs of supporting mathematical lemmas are available at
\url{https://github.com/ExampleOne/unitary-rql-lean}.


%% file: space-complexity-map.tex
\begingroup
\hypersetup{hidelinks}
\captionsetup{font=footnotesize}
\tikzset{
  space node/.style={draw=circuitInk!60,fill=black!3,rounded corners=3pt,
    line width=.65pt,inner xsep=6pt,inner ysep=5pt,align=center,
    font=\small,text=circuitInk},
  space quantum/.style={space node,draw=circuitSource,fill=circuitSource!5},
  space edge/.style={-{Stealth[length=4.5pt,width=3.5pt]},
    draw=circuitInk!85,line width=.7pt},
  space cite/.style={font=\scriptsize,fill=white,inner sep=2pt,
    text=circuitInk,align=center},
  space note/.style={font=\footnotesize,align=center,text=circuitInk},
  space crossing/.style={preaction={draw=white,line width=3.4pt}}
}

\begin{figure}[!htbp]
\centering
\begin{tikzpicture}[x=.97cm,y=.80cm]
  \node[space node] (upper) at (5.7,11.4)
    {$\mathsf P\cap\DSPACE(O(\log^2 n))$};
  \node[space quantum] (pl) at (5.7,9.6)
    {$\PL=\mathsf{PrQL}$\quad{\scriptsize\cite{watrous1999space,Watrous03}}\\[2pt]
     $=\mathsf{PostBQL}=\mathsf{PostBQ_UL}$\\[-1pt]
     {\scriptsize\cite{LeGallNishimuraYakaryilmaz21,Tani22}}};
  \node[space quantum] (nql) at (2.1,7.35)
    {$\mathsf{NQL}=\mathsf{NQ_UL}$\\
     $=\mathsf{coC_{=}L}$\\[-1pt]
     {\scriptsize\cite{watrous1999space,FeffermanRemscrim21}}};
  \node[space quantum] (bql) at (9.0,7.35)
    {$\BQL=\BQUL$\\
     $=\mathsf{Q_UMAL}$\\[-1pt]
     {\scriptsize\cite{FeffermanRemscrim21}}};
  \node[space node] (nl) at (1.45,5.0)
    {$\mathsf{NL}=\mathsf{coNL}$\\[-1pt]
     {\scriptsize\cite{Immerman88,Szelepcsenyi88}}};
  \node[space quantum] (qma) at (9.0,5.25)
    {$\mathsf{coQ_UMAL}_1$};
  \node[space node] (bpl) at (4.0,3.45) {$\BPL$};
  \node[space quantum] (rql) at (8.3,3.45) {$\RQL$};
  \node[space node] (rl) at (2.8,1.6) {$\RL$};
  \node[space quantum] (rqul) at (8.3,1.5)
    {$\RQUL=\mathsf{RQ_UMAL}$\\[-1pt]
     {\scriptsize\cite{FeffermanRemscrim21}}};
  \node[space node] (l) at (5.55,0.0) {$\Lclass$};

  \draw[space edge] (l) -- (rl);
  \draw[space edge] (l) -- (rqul);
  \draw[space edge] (rl) -- (bpl);
  \draw[space edge] (rl) -- (nl);
  \draw[space edge] (rl) -- (rql);
  \draw[space edge] (rqul) -- (rql);
  \draw[space edge] (nl) -- (nql);
  \draw[space edge] (rql) -- (nql);
  \draw[space edge,space crossing] (bpl) -- (bql);
  \draw[space edge] (rql) -- node[space cite,right=3pt]
    {\cite{FeffermanRemscrim21}} (qma);
  \draw[space edge] (qma) -- node[space cite,right=3pt]
    {\cite{FeffermanRemscrim21}$^{\dagger}$} (bql);
  \draw[space edge] (nql) -- (pl);
  \draw[space edge] (bql) -- (pl);
  \draw[space edge] (pl) -- node[space cite,right=3pt]
    {\cite{BorodinCookPippenger83}} (upper);

  \node[space note,text width=11.7cm] at (5.7,-1.35) {
    Exact decision: $\Lclass\subseteq\mathsf{EQ_UL}\subseteq\mathsf{EQL}
       \subseteq\RQL\cap\mathsf{coRQL}$;\\[2pt]
    $\mathsf{EQ_UL}\subseteq\RQUL\cap\mathsf{coRQ_UL}$.
    These are definitional containments.};
\end{tikzpicture}
\caption{\textbf{Ordinary logarithmic space.}
An arrow $A\to B$ means $A\subseteq B$, not necessarily strict.
Unlabelled arrows follow from the definitions and displayed equalities;
crossing lines are not junctions. Transitive edges and some classes are omitted.
The subscript $U$ means unitary until the final measurement.
$R$ means bounded error with perfect soundness, $N$ means positive acceptance
on yes-instances and zero on no-instances, and $\mathsf{PrQL}$ means
two-sided unbounded error (often denoted $\QL$).
$\mathsf{Post}$ allows postselection; $\mathsf{C_{=}L}$ is the zero-test class for
$\mathsf{GapL}$ functions. $\mathsf{Q_UMAL}_1$ has perfect completeness.
$\dagger$: use $\mathsf{Q_UMAL}_1\subseteq\mathsf{Q_UMAL}=\BQL$
and closure of $\BQL$ under complement.
We use polynomial-time, logspace-uniform models and the standard algebraic
gate conventions of the cited results.
In particular, no arbitrary-gate-set invariance for exact or one-sided error is claimed.
This figure summarizes prior results. The reverse inclusion
$\RQL_{\mathcal G}\subseteq\RQUL_{\mathcal G}$ proved in
Theorems~\ref{thm:rql-adjoints} and~\ref{thm:rql-cm} holds under
their stated gate assumptions and is not added to this prior-work diagram.
The layout follows the containment-map convention of Watrous~\cite[Fig.~9]{watrous1999space}.}
\label{fig:logspace-map}
\end{figure}


\endgroup

%% file: intro-main-results.tex
\subsection{Main results}
\label{sec:main-results}
Put $\Gamma=\{H,T,\mathrm{CNOT}\}$, where
$T=\operatorname{diag}(1,e^{i\pi/4})$.
For the two RQL results below, we use fixed finite sets of
one- and two-qubit gates;
all computations have polynomial time and logarithmic classical and
quantum workspace, with every branch on every input halting within a
common polynomial bound. An exact gate implementation is a fixed finite
circuit whose extra qubits start and end in zero on every input. Unless
we explicitly allow a global phase, its action is exactly the specified
matrix $U$, rather than a scalar multiple $e^{i\theta}U$.
We write $M^\dagger=\overline M^{\mathsf T}$ for the Hermitian conjugate
(conjugate transpose) of a matrix. For a unitary gate,
$U^\dagger=U^{-1}$ is its inverse.
\ifStandardGateLead
For the standard gate set, Theorem~\ref{thm:rql-standard} proves
\[
  \boxed{\RQL_\Gamma=\RQUL_\Gamma.}
\]
Thus an adaptive Clifford+$T$ computation with one-sided bounded error
can be simulated with only a final measurement, preserving polynomial
time, logarithmic space, and exactly zero acceptance on no-instances.
The two theorems below extend this conclusion to broader gate sets.
\else
We prove two complementary statements.
\fi

\begin{enumerate}
 \item \textbf{Gate sets with exact inverses
 (Theorem~\ref{thm:rql-adjoints}).}
 Suppose the entries of every gate in $\mathcal G$ can be approximated
 to error $2^{-b}$ in deterministic space $O(b)$ and time $2^{O(b)}$.
 This condition concerns computing numerical descriptions of the gates;
 the primitive gates themselves are applied exactly.
 If circuits over $\mathcal G$ implement every gate of $\Gamma$ and
 every inverse $U^\dagger$, for $U\in\mathcal G$, exactly in this sense,
 then
 \[
   \boxed{\RQL_{\mathcal G}=\RQUL_{\mathcal G}.}
 \]
 Exact conjugates and controlled versions of the additional gates are
 unnecessary. In particular, this covers
 $\Gamma\cup\{U_1,U_1^\dagger,\ldots,U_k,U_k^\dagger\}$ for any
 fixed family whose entries satisfy these bounds, even when they are
 transcendental.
 This statement eliminates measurements within the given gate set;
 it does not identify the classes defined by different gate sets.

 \item \textbf{CM-field gate-set independence
 (Theorem~\ref{thm:rql-cm}).}
 A \emph{number field} is a field obtained from $\mathbb Q$ by adjoining
 finitely many algebraic numbers. It is \emph{totally real} if every
 element has only real algebraic conjugates, meaning that all roots of
 its minimal polynomial over $\mathbb Q$ are real. For example,
 $\mathbb Q(\sqrt2)$ is totally real. A \emph{CM field} has the form
 $K(\sqrt{-a})$, where $K$ is totally real and every conjugate of
 $a\in K$ is positive. For example, $\mathbb Q(\sqrt2,i)$ is a CM field
 containing the Clifford+$T$ entries. Section~\ref{sec:rql-algebraic}
 gives the equivalent formal definition.
 Suppose that, after multiplying each primitive gate by its own fixed
 phase, all entries lie in one fixed CM field. If circuits over
 $\mathcal G$ implement the gates of $\Gamma$ exactly with extra qubits
 returned to zero, now allowing a global phase, then
 \[
   \boxed{\RQL_{\mathcal G}=\RQUL_{\mathcal G}
       =\RQL_{\Gamma}=\RQUL_{\Gamma}.}
 \]
 The additional gates need not have exact implementations of their
 inverses, complex conjugates, or controlled versions.
 The field and all phases are fixed
 independently of input length.
\end{enumerate}
The first result places no algebraic restriction on the entries but requires exact
inverses; the second removes that requirement under a field restriction
and identifies the resulting class with the Clifford+$T$ class. Both preserve zero acceptance
\emph{exactly}, despite approximation in the circuits implementing the
right Kraus factors.
The proofs use the established two-sided measurement-elimination theorem
to implement the right Kraus factors. Accordingly, the bounded-error application
in Appendix~\ref{sec:doubling-bql}
 is conditional on that implementation
 and is not an independent proof of $\BQL=\BQUL$.

\paragraph{Perfect completeness with quantum witnesses
(Theorem~\ref{thm:qmal-one}).}
There is also an exact history-checking construction for quantum
Merlin--Arthur verification. Here Merlin supplies an arbitrary quantum
witness, and its size counts towards the total logarithmic space bound.
For any fixed finite gate set $\mathcal G$ of constant arity whose
entries satisfy the bounds in item~1, implementing $H$, $X$, and Toffoli
exactly, and each primitive inverse
up to a global phase, we prove
\[
 \boxed{\mathsf{QMAL}_{1,\mathcal G}
       =\mathsf{Q_UMAL}_{1,\mathcal G}.}
\]
The subscript $1$ means that some witness is accepted with certainty on
each yes-instance; every witness is accepted with probability at most
$1/3$ on no-instances. All fixed gate implementations return their
extra qubits to zero. Section~\ref{sec:qmal-one} defines the models
and proves the theorem. The construction propagates the rejection observable backwards
and tests a history certificate supplied by Merlin. Paired calls to a
gate and its inverse give the necessary doubled coefficients exactly,
without controlled or conjugated versions of the gate. This construction
requires neither reweighting nor the two-sided measurement-elimination
theorem, but retains a quantum witness.

%% file: figure-rql-roadmap.tex
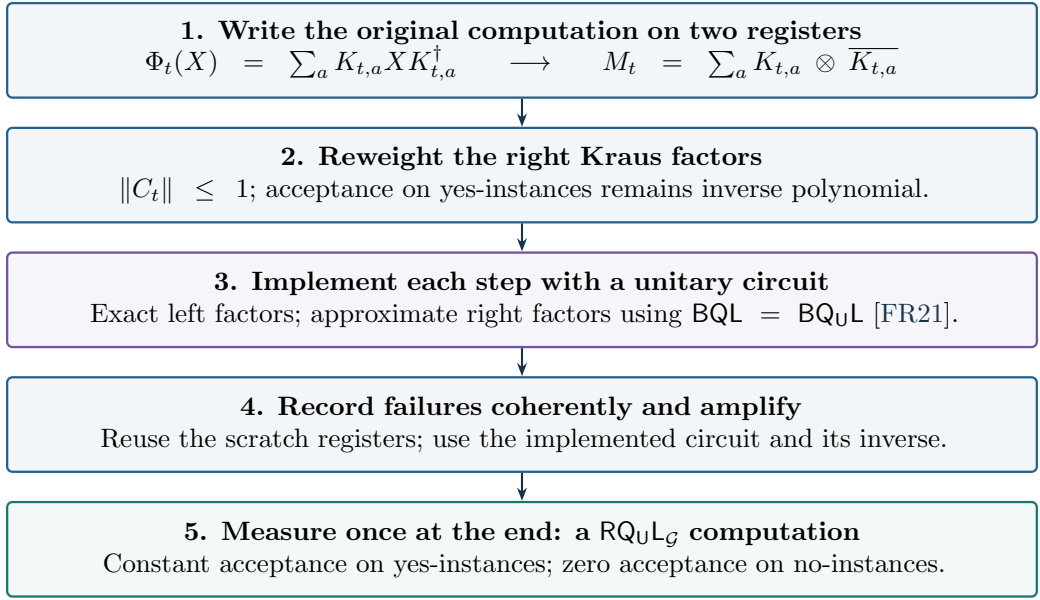
\begin{figure}[tbp]
  \centering
  \begin{tikzpicture}[x=1cm,y=1cm,
    roadmap/.style={circuit gate,text width=13.2cm,minimum height=1.25cm,
      inner sep=6pt,font=\small},
    route/.style={-{Stealth[length=5pt]},draw=circuitInk,line width=.8pt}]
    \node[roadmap] (double) at (0,0)
      {\textbf{1. Write the original computation on two registers}\\
       $\Phi_t(X)=\sum_a K_{t,a}XK_{t,a}^{\dagger}
       \quad\longrightarrow\quad M_t=\sum_a K_{t,a}\otimes\overline{K_{t,a}}$};
    \node[roadmap] (weight) at (0,-1.65)
      {\textbf{2. Reweight the right Kraus factors}\\
       $\|C_t\|\leq1$; acceptance on yes-instances remains inverse polynomial.};
    \node[roadmap,draw=circuitAux,fill=circuitAux!6] (compile) at (0,-3.3)
      {\textbf{3. Implement each step with a unitary circuit}\\
       Exact left factors; approximate right factors using
       $\BQL=\BQUL$~\cite{FeffermanRemscrim21}.};
    \node[roadmap] (compose) at (0,-4.95)
      {\textbf{4. Record failures coherently and amplify}\\
       Reuse the scratch registers; use the implemented circuit and its inverse.};
    \node[roadmap,draw=circuitClean,fill=circuitClean!5] (finish) at (0,-6.6)
      {\textbf{5. Measure once at the end: a $\RQUL_{\mathcal G}$ computation}\\
       Constant acceptance on yes-instances; zero acceptance on no-instances.};
    \foreach \a/\b in {double/weight,weight/compile,compile/compose,compose/finish}
      \draw[route] (\a.south)--(\b.north);
  \end{tikzpicture}
  \caption{The direct RQL simulation under the gate assumptions of
Theorem~\ref{thm:rql-standard} and its extension in Theorem~\ref{thm:rql-adjoints}.
The first two stages describe matrices;
    the last three implement the simulation in logarithmic space and
    polynomial time. Exact left Kraus factors preserve perfect soundness
    throughout, while reweighting and accurate right factors retain
    acceptance on yes-instances. The use of two-sided measurement
    elimination is a dependency of the construction.
    The CM-field theorem first converts the source computation to the
    standard gate set, then applies this simulation.}
  \label{fig:rql-roadmap}
\end{figure}

%% file: prelims.tex
\section{Preliminaries}
\input{qtm}

We define the following complexity classes.

Throughout, a machine solves a promise problem $(L_{\rm yes},L_{\rm no})$
with one-sided error if its acceptance probability is at least $2/3$ on
$L_{\rm yes}$ and is exactly zero on $L_{\rm no}$.
All machines below are \emph{absolutely halting}: there is a polynomial
$P$ such that, for every input $x$, every computation branch halts within
$P(|x|)$ steps. This requirement also applies outside the promise, where
no correctness condition is imposed.

\begin{definition}
  The class $\RQL_{\mathcal{G}}$ consists of promise problems solved with
  the one-sided error specified above by an absolutely halting,
  polynomial-time $\QTM_{\mathcal{G}}$ using $O(\log n)$ qubits and
  $O(\log n)$ classical work bits.
\end{definition}

\begin{definition}
  The class $\RQUL_{\mathcal{G}}$ is the class of promise problems that can be solved
  with the one-sided error specified above by an absolutely halting,
  polynomial-time $\UQTM_{\mathcal{G}}$ using $O(\log n)$ qubits and
  $O(\log n)$ classical work bits.
\end{definition}


%% file: qtm.tex
\subsection{Quantum Turing machines}

We use the \emph{classically controlled} model of space-bounded quantum computation, in which a classical Turing machine directs operations on a tape of qubits~\cite{TaShma13,FeffermanRemscrim21}.
The cited works use this model. Watrous' earlier model allows the
entire machine configuration to be in superposition~\cite{watrous1999space}.

\begin{definition}[Quantum Turing machine over $\mathcal{G}$]
\label{def:qtm}
Let $\mathcal{G}$ be a finite set of one- and two-qubit unitaries, for example $\mathcal{G} = \{H, T, \mathrm{CNOT}\}$.
A \emph{quantum Turing machine over $\mathcal{G}$}, or $\QTM_{\mathcal{G}}$,
is a deterministic Turing machine $M$ with a finite set of states $Q$,
a read-only input tape, and a classical work tape. It also has a
\emph{quantum work tape}, scanned by two heads. Each cell of this tape
holds a qubit initialized to $\ket0$.
The set $Q$ contains a halting accepting state $q_{\mathrm{acc}}$, a halting rejecting state $q_{\mathrm{rej}}$, a state $q_g$ for each gate $g \in \mathcal{G}$, and a measurement state $q_{\mathrm{meas}}$.
In each step $M$ behaves as an ordinary Turing machine, except that:
\begin{itemize}
  \item on entering $q_g$, the gate $g$ is applied to the qubit(s) under the quantum heads (a one-qubit gate acts on the qubit under the first head);
  \item on entering $q_{\mathrm{meas}}$, the qubit under the first quantum head is measured in the computational basis, and the outcome $b \in \{0,1\}$ is made available to the finite control, which determines the next step.
\end{itemize}
Measurement outcomes are the only way information flows from the quantum tape to the classical control.
$M$ \emph{accepts} if it halts in $q_{\mathrm{acc}}$; the acceptance probability $\Pr[M(x)\text{ accepts}]$ is taken over the measurement outcomes.

$M$ runs in \emph{space} $s(n)$ if on every input of length $n$ and every sequence of measurement outcomes it visits at most $s(n)$ cells of each work tape.
We require \emph{absolute halting}: every computation branch halts.
A time bound $t(n)$ means that every branch on every input of length $n$
halts within $t(n)$ steps, including inputs outside the promise.
Thus time is bounded in the worst case over both inputs and measurement
outcomes. Throughout this paper, space-$s$ machines satisfy such a bound
$t(n)=2^{O(s(n))}$; for logarithmic space this is polynomial.
\end{definition}

We distinguish two kinds of $\QTM_{\mathcal{G}}$ according to how they may use measurements.

\begin{definition}[General and unitary QTMs]
\label{def:general-unitary-qtm}
\leavevmode
\begin{enumerate}
  \item A \emph{general} $\QTM_{\mathcal{G}}$ may enter $q_{\mathrm{meas}}$ any number of times, so its subsequent classical and quantum operations may depend on intermediate measurement outcomes.
  \item A \emph{unitary} $\QTM_{\mathcal{G}}$, or $\UQTM_{\mathcal{G}}$, enters $q_{\mathrm{meas}}$ exactly once, as the last step of its computation: it measures a designated output qubit and then immediately halts, in $q_{\mathrm{acc}}$ if the outcome is $1$ and in $q_{\mathrm{rej}}$ otherwise.
\end{enumerate}
Every $\UQTM_{\mathcal{G}}$ is in particular a general $\QTM_{\mathcal{G}}$.
When we write $\QTM_{\mathcal{G}}$ without qualification we mean a general one.
\end{definition}

The classical control of a $\UQTM_{\mathcal{G}}$ never depends on the quantum tape, so on input $x$ it deterministically applies a fixed sequence of $2^{O(s)}$ gates to $O(s)$ qubits, followed by a single measurement.
In other words, space-$s$ $\UQTM_{\mathcal{G}}$s are equivalent to $\DSPACE(s)$-uniform families of unitary circuits of width $O(s)$ and size $2^{O(s)}$~\cite{FeffermanLin18,FeffermanRemscrim21}.
Likewise, general $\QTM_{\mathcal{G}}$s correspond to uniform circuits with intermediate measurements whose outcomes classically control later gates.
Under the usual universality and gate-computability assumptions, the
space-efficient Solovay--Kitaev theorem makes two-sided bounded-error
classes such as $\BQL$ independent of the gate set~\cite{FeffermanRemscrim21}.
This argument does not extend directly to one-sided error, since approximating gates can make a rejecting computation accept with small nonzero probability, which is why we make the gate set explicit.

%% file: result.tex
\section{Removing measurements from \texorpdfstring{$\RQL$}{RQL} computations}

\input{result-doubling}

\input{result-standard}
\input{result-adjoints}

\input{result-algebraic}

%% file: result-doubling.tex
\subsection{The reweighted doubling framework}
\label{sec:doubling-framework}

We first develop the linear-algebraic construction shared by the bounded-
and one-sided-error simulations. Let $T\ge1$, let $d$ be the dimension
of the source configuration register, and, for $1\le t\le T$, write
\begin{equation}\label{eq:doubling-source}
 \rho_0=\ket0\bra0,\qquad
 \rho_t=\Phi_t(\rho_{t-1}),\qquad
 \Phi_t(X)=\sum_a K_{t,a}XK_{t,a}^\dagger.
\end{equation}
Each $\Phi_t$ is trace preserving. The source acceptance probability is
$p=\Tr(\Pi\rho_T)$ for a computational-basis acceptance projector $\Pi$.
For logarithmic-space, polynomial-time computations, the register
includes the retained classical configuration and quantum workspace,
and $d,T\le n^{O(1)}$.

\paragraph{Vectorization and normalization.}
Use row vectorization
\begin{gather*}
 \vecf(X)=\sum_{i,j}X_{ij}\ket i\ket j,
 \qquad
 M_t=\sum_a K_{t,a}\otimes\overline{K_{t,a}},\\
 \vecf(\rho_t)=M_t\vecf(\rho_{t-1}).
\end{gather*}
This doubles the number of qubits used to index a matrix. It is a
linear representation of $\rho$, rather than a purification of it:
$\norm{\vecf(\rho)}^2=\Tr(\rho^2)$, and its normalized first-register reduced state
is $\rho^2/\Tr(\rho^2)$. The vector $\vecf(\sqrt\rho)$ is the
canonical purification.

The matrices $M_t$ need not be contractions. For example, resetting a
qubit gives $M=\ket{00}(\bra{00}+\bra{11})$, with operator norm
$\sqrt2$. Normalizing every step separately to obtain a unitary clean
block can lose an exponential factor over $T$ steps. This loss concerns
the implementation normalization; the density-vector norm itself is
always at least $1/\sqrt d$.

\paragraph{Reference states and contractive steps.}
We start the reference computation in the maximally mixed state.
At each step, we apply the original channel and then mix in a small
maximally mixed contribution. Following the channel adapts the reference
to the computation; the added contribution keeps every eigenvalue
bounded away from zero, so the inverse square roots below are controlled.
Set
\begin{equation}\label{eq:doubling-reference}
 \eta=1/T,\qquad \kappa=1+\eta,\qquad \sigma_0=I/d,
 \qquad
 \sigma_t=\frac{\Phi_t(\sigma_{t-1})+\eta I/d}{\kappa}.
\end{equation}
These are density operators satisfying
$\sigma_t\succeq I/(d(T+1))$; every eigenvalue is at least
$1/(d(T+1))$. Define
\begin{align}
 J_{t,a}&=\kappa^{-1/2}\sigma_{t-1}^{1/2}
             K_{t,a}^\dagger\sigma_t^{-1/2},
 &J_{t,*}&=\sqrt{\frac{\eta}{\kappa d}}\,\sigma_t^{-1/2},\nonumber\\
 W_t\ket v&=\sum_a J_{t,a}\ket v\ket a+J_{t,*}\ket v\ket *,
 &C_t&=\sum_a K_{t,a}\otimes\overline{J_{t,a}^\dagger}.
 \label{eq:doubling-operators}
\end{align}
The reference recurrence gives
\[
 W_t^\dagger W_t
 =\kappa^{-1}\sigma_t^{-1/2}
    \bigl(\Phi_t(\sigma_{t-1})+\eta I/d\bigr)\sigma_t^{-1/2}
 =I.
\]
Thus $W_t$ is an isometry: it preserves vector lengths. Choose a unitary dilation $D_t$ of
the Kraus operators and an abstract unitary extension
$U_{\overline W_t}$ of $\overline W_t$, padding the label space as
needed and setting $K_{t,*}=0$. On doubled registers $A,B$ and label
register $E$,
\[
 C_t=\bra0_E (U_{\overline W_t}^{\dagger})_{BE}(D_t)_{AE}\ket0_E.
\]
It is therefore a compression of a unitary, and $\norm{C_t}\le1$ in
Euclidean operator norm. This identity does not assume exact synthesis
of the abstract extension over a finite gate set.

\paragraph{Keeping the left factors exact.}
Trace preservation ensures that the original left preparation preserves
vector lengths: $\sum_a K_{t,a}^{\dagger}K_{t,a}=I$.
The preparation for the unweighted right factor would use Hermitian
conjugates, which need not satisfy the same condition:
$\sum_a K_{t,a}K_{t,a}^{\dagger}=\Phi_t(I)$ may differ from $I$.
The factors $J_{t,a}$ and the extra branch $J_{t,*}$ correct this
second normalization, as the identity $W_t^{\dagger}W_t=I$ shows.
Figure~\ref{fig:kraus-factors}(a) places both reference-state weights
on the right register. In the implementation of
Section~\ref{sec:rql-standard}, the original circuit $D_t$ remains exact
while the circuit for the right factors is approximated, as shown in
Figure~\ref{fig:kraus-factors}(b).
\input{circuit-kraus-factors}


\paragraph{Cancellation of weights and acceptance bounds.}
Vectorization gives
\[
 C_t\vecf(X)=\kappa^{-1/2}\vecf\!\left(
   \Phi_t(X\sigma_{t-1}^{1/2})\sigma_t^{-1/2}\right).
\]
Consecutive square roots cancel, yielding
\begin{equation}\label{eq:doubling-endpoint}
 C_T\cdots C_1\vecf(\rho_0)
 =\frac{\vecf(\rho_T\sigma_T^{-1/2})}{\sqrt{d\kappa^T}}.
\end{equation}
Let
\begin{equation}\label{eq:doubling-signal}
 q_*:=\norm{(\Pi\otimes I)C_T\cdots C_1\ket{00}}^2
     =\frac{\norm{\Pi\rho_T\sigma_T^{-1/2}}_F^2}{d\kappa^T}.
\end{equation}
In the unitary implementation, this is the joint probability that the
first register indicates acceptance and the scratch registers and
failure counter are zero. This probability satisfies
\begin{equation}\label{eq:doubling-signal-bounds}
 \frac{p^2}{ed}\le\frac{p^2}{d\kappa^T}\le q_*\le p,
 \qquad e=\exp(1).
\end{equation}
For the lower bound, use the Frobenius inner product:
\[
 \bigl\langle\Pi\sigma_T^{1/2},
     \Pi\rho_T\sigma_T^{-1/2}\bigr\rangle_F=p,
 \qquad
 \norm{\Pi\sigma_T^{1/2}}_F^2=\Tr(\Pi\sigma_T)\le1.
\]
Cauchy--Schwarz and $(1+1/T)^T\le e$ give the result.

For the upper bound, positivity of the channels and $\sigma_0\succeq
\rho_0/d$ imply $\sigma_T\succeq\rho_T/(d\kappa^T)$.
Put $a=d\kappa^T$ and $X=\sigma_T^{-1/2}\rho_T\sigma_T^{-1/2}$.
Then $0\preceq X\preceq aI$, so $X^2\preceq aX$. Congruence by
$\sigma_T^{1/2}$ gives
$\rho_T\sigma_T^{-1}\rho_T\preceq a\rho_T$.
Taking the trace against $\Pi$ proves $q_*\le p$.
Neither bound requires the operators to commute.

\paragraph{Example: repeated qubit resets.}
Figure~\ref{fig:reset-reweighting} makes the normalization issue
explicit for $\Phi(X)=\Tr(X)\ket0\bra0$.
Starting in $\ket0$ and accepting $\ket0$ after $T$ resets gives $p=1$.
Dividing each unweighted doubled step by $\sqrt2$ leaves acceptance
probability $2^{-T}$. In contrast, our reference recurrence gives
$q_*=[(2+1/T)(1+1/T)^{T-1}]^{-1}\geq1/(2e)$.
Indeed, $\sigma_t=\operatorname{diag}(2+\eta,\eta)/(2\kappa)$ for
every $t\geq1$, so $C_1\ket{00}=\ket{00}/\sqrt{2+\eta}$ and
$C_t\ket{00}=\kappa^{-1/2}\ket{00}$ for $t\geq2$.
The stated lower bound follows from $2+\eta\leq2\kappa$ and
$\kappa^T\leq e$.
\input{figure-reset-reweighting}

\input{recovery-map-connection}

Section~\ref{sec:rql-standard} implements the approximate clean blocks
coherently in logarithmic space while keeping the left Kraus factors
exact. Each accepting history remains zero on no-instances, and
unitary amplification raises the inverse-polynomial acceptance probability to a constant.
Appendix~\ref{sec:doubling-bql} records the separate conditional
bounded-error application.


%% file: circuit-kraus-factors.tex
\begin{figure}[tbp]
  \centering
  {\small\textbf{(a) Matrix factors for one Kraus label $a$}}\par\smallskip
  \begin{tikzpicture}[x=1cm,y=.85cm]
    \foreach \y in {0,1.35}{\draw[circuit wire] (0,\y)--(11.5,\y);}
    \node[circuit label,anchor=east] at (-.15,1.35) {left $A$};
    \node[circuit label,anchor=east] at (-.15,0) {right $B$};
    \node[circuit gate,minimum width=1.6cm,minimum height=.7cm]
      at (5.75,1.35) {$K_{t,a}$};
    \node[circuit auxiliary,minimum width=2cm,minimum height=.8cm]
      at (1.6,0) {$\overline{\sigma_{t-1}^{1/2}}$};
    \node[circuit auxiliary,minimum width=1.6cm,minimum height=.8cm]
      at (5.75,0) {$\overline{K_{t,a}}$};
    \node[circuit auxiliary,minimum width=2cm,minimum height=.8cm]
      at (9.9,0) {$\overline{\sigma_t^{-1/2}}$};
    \node[circuit label] at (5.75,-.92)
      {Common scalar factor: $\kappa^{-1/2}$};
  \end{tikzpicture}
  \[
    C_t=\sum_a K_{t,a}\otimes\overline{J_{t,a}^{\dagger}},
    \qquad
    \overline{J_{t,a}^{\dagger}}
      =\kappa^{-1/2}\overline{\sigma_t^{-1/2}}\,
        \overline{K_{t,a}}\,\overline{\sigma_{t-1}^{1/2}}.
  \]
  {\small\textbf{(b) Unitary implementation with approximate right factors}}\par\smallskip
  \begin{tikzpicture}[x=1cm,y=.8cm]
    \foreach \y in {0,1.2,2.4,3.6}{\draw[circuit wire] (0,\y)--(11.5,\y);}
    \node[circuit label,anchor=east] at (-.15,3.6) {left $A$};
    \node[circuit label,anchor=east] at (-.15,2.4) {$E:\ket0$};
    \node[circuit label,anchor=east] at (-.15,1.2) {right $B$};
    \node[circuit label,anchor=east] at (-.15,0) {$Z:\ket0$};
    \node[circuit gate,minimum width=1.9cm,minimum height=1.65cm]
      at (2,3) {$D_t$\\\small exact};
    \node[circuit auxiliary,minimum width=2.6cm,minimum height=2.7cm]
      at (6.25,1.2) {$Q_t^{-1}$\\\small approximate\\\small right factors};
    \node[circuit tag] at (9.8,2.4) {$E=0$};
    \node[circuit tag] at (9.8,0) {$Z=0$};
    \node[circuit label,anchor=west] at (11.65,3.6) {$A$};
    \node[circuit label,anchor=west] at (11.65,1.2) {$B$};
  \end{tikzpicture}
  \[
    \widehat C_t
      =\bra{0_E0_Z}(Q_t^{-1})_{BEZ}(D_t)_{AE}\ket{0_E0_Z}
      =\sum_a K_{t,a}\otimes B_{t,a}.
  \]
  \caption{Keeping the left Kraus factors exact while approximating the
    right factors. Read both diagrams from left to right.
    Panel (a) is an algebraic factorization: its matrix boxes need not be
    unitary. Both reference-state weights act on $B$, and summing over
    $a$ gives $C_t$. Panel (b) is a unitary circuit.
    The exact source circuit $D_t$ writes the Kraus label into $E$;
    $Q_t$ acts only on $B,E,Z$ and supplies the approximate right factors.
    Its implemented inverse is exact. In particular,
    $B_{t,a}=\bra{0_E0_Z}Q_t^{-1}\ket{a_E0_Z}$, so approximation
    leaves every left factor $K_{t,a}$ unchanged.
    The green conditions select the component with zero scratch;
    they are not intermediate measurements. The full simulation uses
    the failure counter described in Section~\ref{sec:rql-adjoints}.
    Appendix~\ref{app:implementing-doubling} gives the detailed circuits.}
  \label{fig:kraus-factors}
\end{figure}
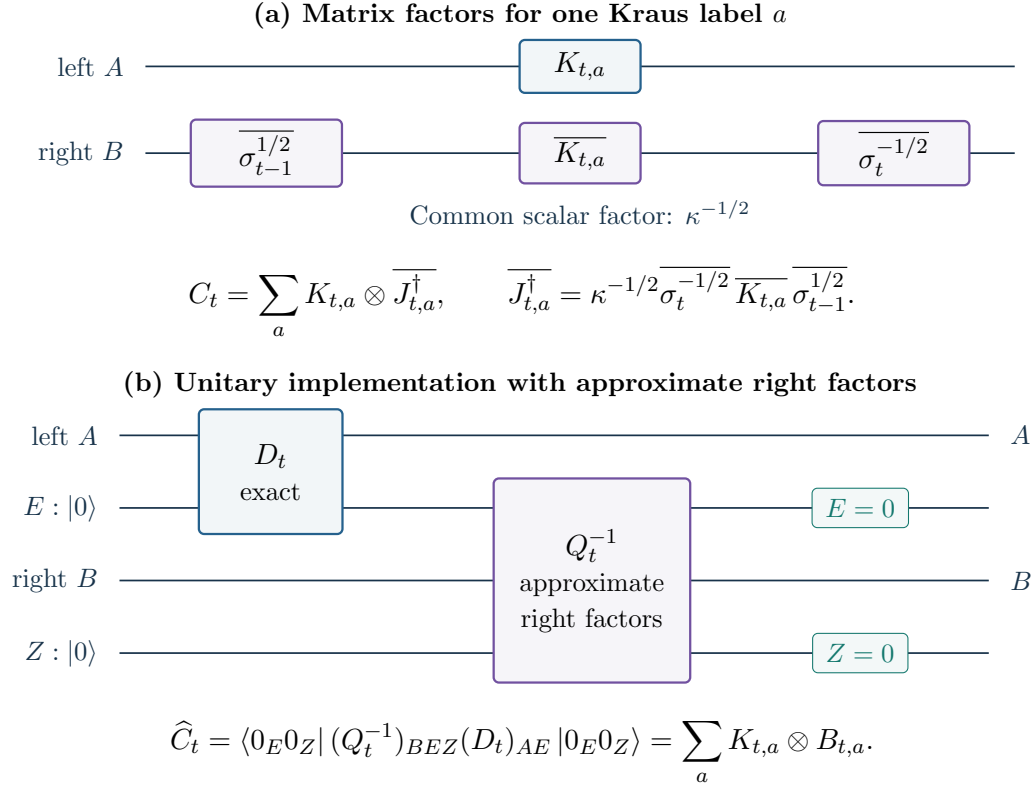

%% file: figure-reset-reweighting.tex
\begin{figure}[tbp]
  \centering
  {\small\textbf{(a) One reset: growth before reweighting, contraction after it}}\par\smallskip
  \begin{tikzpicture}[x=1cm,y=1cm]
    \node[circuit label] at (.4,0) {$\ket\Omega$};
    \node[circuit label] at (.4,-1.2) {$\ket\Omega$};
    \foreach \y in {0,-1.2}{
      \draw[-{Stealth[length=4.5pt]},circuit wire] (1.1,\y)--(2.25,\y);
      \draw[-{Stealth[length=4.5pt]},circuit wire] (6.4,\y)--(8.0,\y);
    }
    \node[circuit gate,minimum width=4cm,minimum height=.72cm] at (4.3,0)
      {$M=\ket{00}(\bra{00}+\bra{11})$};
    \node[circuit auxiliary,minimum width=4cm,minimum height=.72cm] at (4.3,-1.2)
      {$C_1=M/\sqrt{2+\eta}$};
    \node[circuit label,anchor=west] at (8.1,0)
      {$\sqrt2\ket{00}\qquad\|M\|=\sqrt2$};
    \node[circuit label,anchor=west] at (8.1,-1.2)
      {$\sqrt{\dfrac{2}{2+\eta}}\ket{00}\qquad\|C_1\|<1$};
    \node[circuit label] at (6.2,-2.02)
      {$\ket\Omega=(\ket{00}+\ket{11})/\sqrt2$ has unit length.};
  \end{tikzpicture}
  \par\medskip
  {\small\textbf{(b) Repeated resets: probability before amplification}}\par\smallskip
  \begin{tikzpicture}[x=.5cm,y=.58cm]
    \foreach \j/\lab in {0/1,2/{10^{-2}},4/{10^{-4}},6/{10^{-6}}}{
      \draw[circuitInk!15] (0,-\j)--(21,-\j);
      \node[circuit label,anchor=east] at (-.2,-\j) {$\lab$};
    }
    \draw[circuit wire] (0,.15)--(0,-6.25)--(21,-6.25);
    \foreach \j in {1,5,10,15,20}{
      \draw[circuit wire] (\j,-6.25)--(\j,-6.38);
      \node[circuit label,anchor=north] at (\j,-6.45) {$\j$};
    }
    \draw[draw=circuitSource,line width=1.1pt,dashed]
      plot[domain=1:20,samples=20] (\x,{-\x*ln(2)/ln(10)});
    \draw[draw=circuitAux,line width=1.2pt]
      plot[domain=1:20,samples=60]
        (\x,{-ln((2+1/\x)*(1+1/\x)^(\x-1))/ln(10)});
    \node[circuit label,anchor=west,text=circuitAux] at (7,-1.35)
      {Reference-state reweighting};
    \node[circuit label,anchor=west,text=circuitSource,fill=white,inner sep=2pt]
      at (1.5,-4.65) {Normalize each step by $\sqrt2$};
    \node[circuit label] at (10.5,-7.5) {Number of resets $T$};
    \node[circuit label,rotate=90] at (-2.5,-3) {Probability (logarithmic scale)};
  \end{tikzpicture}
  \caption{Reweighting a qubit reset
    $\Phi(X)=\operatorname{Tr}(X)\ket0\bra0$.
    In (a), $\sigma_0=I/2$ and
    $\sigma_1=\operatorname{diag}(2+\eta,\eta)/(2\kappa)$,
    where $\kappa=1+\eta$. Substitution into the definition of $C_1$
    gives the displayed contraction.
    In (b), the source starts in $\ket0$, applies $T$ resets, and accepts
    $\ket0$ with certainty. Normalizing each doubled step as $M/\sqrt2$
    gives probability $2^{-T}$. With $\eta=1/T$, the reference state is
    the same after every reset; for $t\geq2$,
    $C_t\ket{00}=\kappa^{-1/2}\ket{00}$. Thus reweighting gives
    $q_*=[(2+1/T)(1+1/T)^{T-1}]^{-1}$, tending to $1/(2e)$.
    The curves compare exact matrix products before compilation and
    amplification. The vector $\ket\Omega$ in (a) tests the operator
    norm; it is not the initial vector used in (b).}
  \label{fig:reset-reweighting}
\end{figure}

%% file: recovery-map-connection.tex
\paragraph{Relation to Petz recovery.}
The square-root reweighting has a precise connection with the Petz
(or transpose) recovery map~\cite{Petz86,Petz88,BarnumKnill02}.
For a channel $\mathcal N$ and full-rank reference states
$\sigma$ and $\tau=\mathcal N(\sigma)$, this map is
\begin{equation}\label{eq:petz-map}
 \mathcal R_{\sigma,\mathcal N}(Y)
 =\sigma^{1/2}\mathcal N^\dagger
   (\tau^{-1/2}Y\tau^{-1/2})\sigma^{1/2},
\end{equation}
where $\mathcal N^\dagger$ is the Hermitian conjugate of $\mathcal N$,
acting backwards on observables:
$\mathcal N^\dagger(Y)=\sum_b L_b^\dagger YL_b$ for Kraus operators
$L_b$ of $\mathcal N$. The recovery Kraus operators are
$\sigma^{1/2}L_b^\dagger\tau^{-1/2}$.
Their squared products sum to $I$, and
$\mathcal R_{\sigma,\mathcal N}(\tau)=\sigma$.
Recovering arbitrary other states is an additional property, not a
consequence of this identity.

For the references in~\eqref{eq:doubling-reference}, set
\[
 \Psi_t(X)=\frac{\Phi_t(X)+\eta\Tr(X)I/d}{\kappa},
 \qquad \Psi_t(\sigma_{t-1})=\sigma_t.
\]
This is a channel; the trace in its replacement term is essential.
Its Kraus family consists of $K_{t,a}/\sqrt\kappa$ and
$\sqrt{\eta/(\kappa d)}\ket i\bra j$.
The original-channel Kraus operators of
$\mathcal R_{\sigma_{t-1},\Psi_t}$ are exactly the $J_{t,a}$ in
\eqref{eq:doubling-operators}. Its replacement-branch Kraus operators
are instead
\[
 R_{t,ij}=\sqrt{\frac{\eta}{\kappa d}}\,
          \sigma_{t-1}^{1/2}\ket j\bra i\sigma_t^{-1/2}.
\]
Since $\Tr\sigma_{t-1}=1$, they satisfy
\begin{equation}\label{eq:petz-completion}
 \sum_{i,j}R_{t,ij}^\dagger R_{t,ij}
 =\frac{\eta}{\kappa d}\sigma_t^{-1}
 =J_{t,*}^\dagger J_{t,*}.
\end{equation}
Thus our $W_t$ retains the Petz original-channel factors and uses a
single Kraus operator with the same remaining normalization contribution.
It generally induces a \emph{different} channel: the two completion
maps are respectively
\[
 Y\longmapsto \frac{\eta}{\kappa d}
                   \Tr(\sigma_t^{-1}Y)\sigma_{t-1},
 \qquad
 Y\longmapsto \frac{\eta}{\kappa d}
                   \sigma_t^{-1/2}Y\sigma_t^{-1/2}.
\]
One can instead use the full Petz isometry by replacing the label $*$
with the $d^2$ labels $(i,j)$ and assigning all their left Kraus factors
zero. The doubled operator $C_t$ is unchanged. This requires only
$O(\log d)$ additional label qubits, so the total remains $O(\log n)$. Neither version asserts that $C_t$ itself is a
recovery channel: it is a contraction on vectorized matrices.

Gily\'en, Lloyd, Marvian, Quek and Wilde give an algorithm for Petz
recovery using block encodings, quantum singular value transformation,
and oblivious amplitude amplification~\cite{GilyenLloydMarvianQuekWilde22}.
Their isometry construction combines reference square roots, the inverse
of a unitary dilation of the channel, and an inverse square root of the
output reference;
these are also the ingredients of Appendix~\ref{app:aux-synthesis}.
Their access assumptions do not by themselves establish the
logarithmic-space uniform reference access required here, nor exact
preservation of a zero accepting amplitude. We establish those
properties separately.

%% file: result-standard.tex
\subsection{The Clifford+\texorpdfstring{$T$}{T} case}
\label{sec:rql-standard}

We first give the complete simulation for the standard gate set
$\Gamma=\{H,T,\mathrm{CNOT}\}$. This case contains the main argument:
keep the original Kraus operators exact, approximate only the right
factors, and amplify without introducing acceptance on no-instances.
The later subsections extend the available gate sets.

\begin{theorem}[Measurement elimination over Clifford+$T$]
\label{thm:rql-standard}
For absolutely halting, polynomial-time computations with logarithmic
space,
\[
 \RQL_{\Gamma}=\RQUL_{\Gamma}.
\]
\end{theorem}

We use the following implementation result. Its proof in
Appendix~\ref{app:aux-synthesis} uses the established two-sided
measurement-elimination theorem~\cite{FeffermanRemscrim21} and unitary
error reduction~\cite{FeffermanKLMN16}, independently of our one-sided
theorem. Thus the result below isolates the numerical work needed to
implement the reweighted operators.

\begin{lemma}[Uniform implementation of the right Kraus factors]
\label{lem:rql-auxiliary}
Let $T,d\leq n^{O(1)}$. Suppose the source channels have exact unitary
dilations $D_t$, given by polynomial-size, logspace-uniform circuits
over $\Gamma$ on $O(\log n)$ qubits, including their environment.
For each inverse-polynomial $\varepsilon>0$, there are
polynomial-size, logspace-uniform unitary circuits $Q_t$ over $\Gamma$
on $O(\log n)$ qubits such that
\[
 H_t=\sum_a H_{t,a}\otimes\ket a,
 \qquad H_{t,a}=\bra{a_E0_Z}Q_t\ket{0_E0_Z},
 \qquad \norm{H_t-\overline W_t}\leq\varepsilon.
\]
The data register is implicit. The labels include $*$ and any unused
labels in~\eqref{eq:doubling-operators}; there are polynomially many.
The same statement holds with $\Gamma$ replaced throughout by any
fixed finite constant-arity gate set satisfying conditions (1)--(4)
of Appendix~\ref{app:controlled-interface}.
\end{lemma}

\begin{proof}[Proof of Theorem~\ref{thm:rql-standard}]
The inclusion $\RQUL_\Gamma\subseteq\RQL_\Gamma$ is immediate.
For the converse, fix a source machine and an input $x$ of length $n$.
We describe its unitary simulation in five steps.

\paragraph{1. Implement each original step exactly.}
Include the retained classical configuration, head positions and
quantum work tape in the $\log d=O(\log n)$-qubit data register.
Pad unused basis configurations so $d$ is a power of two, and pad
halting configurations with idle steps to a polynomial time bound
$T\geq1$. This gives the channels and acceptance projector of
Section~\ref{sec:doubling-framework}.
Each channel has an exact dilation
\begin{equation}\label{eq:rql-forward}
 D_t(\ket v\ket0_E)=\sum_a K_{t,a}\ket v\ket a_E.
\end{equation}
For a unitary step, apply the gate. For a computational-basis
measurement, copy the measured bit into $E$ using CNOT. For a reset,
copy the old bit into $E$ and XOR that copy back into the data qubit.
Classically controlled steps are handled by copying the old classical
configuration and any measurement outcome into $E$, then reversibly
computing the new configuration. These saved registers ensure that
the update is reversible; tracing them out gives the original channel
on every reachable source state.
Enumerating the polynomially many configurations and tape positions
gives polynomial-size, logspace-uniform circuits. One step needs only
$O(\log n)$ environment qubits.

All gates used here have exact Clifford+$T$ circuits. In particular,
$H^{-1}=H$, $\mathrm{CNOT}^{-1}=\mathrm{CNOT}$ and $T^{-1}=T^7$;
complex-conjugated gates also have exact circuits over $\Gamma$. Toffoli is exact:
the identity
\[
 4abc=a+b+c-(a\mathbin\oplus b)-(a\mathbin\oplus c)
 -(b\mathbin\oplus c)+(a\mathbin\oplus b\mathbin\oplus c)
\]
implements CCZ using parity computations and $T$ gates, after which
conjugating the target by $H$ gives Toffoli. Controlled-$T$ computes the
AND of control and target into zero scratch, applies $T$, and
uncomputes. Controlled-CNOT is Toffoli. Finally, with $S=T^2$ and
$W=SHTHS^{-1}$, the identity $WZW^{-1}=H$ gives an exact controlled-$H$
using controlled-$Z$. Further controls reduce to a reversible AND.
Every extra qubit in these fixed gate circuits is returned to zero.

\paragraph{2. Approximate the right Kraus factors.}
Use the reference states from~\eqref{eq:doubling-reference} and apply
Lemma~\ref{lem:rql-auxiliary} with $\varepsilon=1/(12Td)$.
On doubled data registers $A,B$, first apply $D_t$ to $A,E$, and then
apply the inverse of the actual circuit $Q_t$ to $B,E,Z$.
The block with $E,Z$ initialized and returned to zero is exactly
\begin{equation}\label{eq:rql-actual-step}
 \widehat C_t=\sum_a K_{t,a}\otimes H_{t,a}^{\dagger},
 \qquad \norm{\widehat C_t}\leq1,
 \qquad \norm{\widehat C_t-C_t}\leq\varepsilon.
\end{equation}
Set $K_{t,*}=0$ and use zero left Kraus factors on all padding labels.
The contraction bound holds because this is a block of a unitary.
The error bound follows by composing the difference between the
implemented right block and its target isometry with the exact left
isometry and projecting onto zero scratch. No implementation of an exact matrix square root is
required. Figure~\ref{fig:kraus-factors} shows this separation.

\paragraph{3. Preserve perfect soundness and enough acceptance.}
For a history $h=(a_1,\ldots,a_T)$, write
$K_h=K_{T,a_T}\cdots K_{1,a_1}$. The source probability satisfies
\[
 p=\sum_h\norm{\Pi K_h\ket0}^2.
\]
On a no-instance $p=0$, so every vector $\Pi K_h\ket0$ is zero.
Expanding the product in~\eqref{eq:rql-actual-step} therefore gives
\begin{equation}\label{eq:rql-perfect-soundness}
 (\Pi\otimes I)\widehat C_T\cdots\widehat C_1\ket{00}=0.
\end{equation}
This conclusion holds for any right factors, with no accuracy
assumption; Figure~\ref{fig:perfect-soundness} illustrates why.
\input{figure-perfect-soundness}

For yes-instances, the ideal signal from
Section~\ref{sec:doubling-framework} satisfies
\begin{equation}\label{eq:rql-ideal-success}
 q_*\geq\frac{p^2}{ed},\qquad e=\exp(1).
\end{equation}
Telescoping the products of contractions gives
\[
 \norm{\widehat C_T\cdots\widehat C_1-C_T\cdots C_1}
 \leq T\varepsilon=\frac1{12d}.
\]
Since $p\geq2/3$ and $e<4$, the actual accepting norm is at least
$1/(3\sqrt d)-1/(12d)\geq1/(4\sqrt d)$. Its square is therefore at
least $1/(16d)$.

\paragraph{4. Compose the steps with logarithmic space.}
Let $V_t$ be the full unitary used in step 2 and let $P$ project its
scratch registers onto zero. Reuse these registers at every step.
Add a counter $F$ of modulus $N>T$, where $N$ is a power of two, and
write $X_N\ket f=\ket{f+1\bmod N}$. Apply
\begin{equation}\label{eq:adjoint-counter-step}
 \mathcal V_t=
 \bigl(P\otimes I_F+(I-P)\otimes X_N\bigr)
 (V_t\otimes I_F).
\end{equation}
Thus we apply $V_t$ unconditionally, then increment the counter if
its scratch is nonzero. These operations are unitary and the
increment uses only a computational-basis predicate.
To verify the selected block, put $P_0=P$, $P_1=I-P$ and expand:
\[
 \mathcal V_T\cdots\mathcal V_1
 =\sum_{\delta\in\{0,1\}^T}
 P_{\delta_T}V_T\cdots P_{\delta_1}V_1
 \otimes X_N^{\delta_1+\cdots+\delta_T}.
\]
Starting and ending with counter zero selects exactly the term
$\delta_1=\cdots=\delta_T=0$, since $0\leq\sum_t\delta_t\leq T<N$.
With initial scratch zero, the selected operation is exactly
$\widehat C_T\cdots\widehat C_1$, with final scratch zero.
The final good predicate requires counter zero, scratch zero and
acceptance in $A$. Its probability $\widehat q$ is zero on no-instances
and at least $1/(16d)$ on yes-instances. All failure components remain
in the circuit; no intermediate measurement or postselection occurs.

\paragraph{5. Amplify with exact reflections.}
Append a Hadamard coin and include coin one in the good predicate.
Call the resulting preparation $A$ and its good projector
$P_{\rm good}$. Its success probability is $q=\widehat q/2\leq1/2$.
Use the Grover iterate~\cite{BrassardHoyerMoscaTapp02}
\[
 R_0=2\ket0\bra0-I,\qquad R_{\rm good}=I-2P_{\rm good},
 \qquad Q=A R_0 A^{\dagger}R_{\rm good}.
\]
The inverse $A^{\dagger}$ reverses the actual implemented circuit,
including its approximations. Both reflections are exact Boolean
predicates. Moreover, for an external control,
\begin{equation}\label{eq:adjoint-controlled-grover}
 \operatorname{ctrl}(Q)=
 (I\otimes A)\operatorname{ctrl}(R_0)
 (I\otimes A^{\dagger})\operatorname{ctrl}(R_{\rm good}).
\end{equation}
On control zero the preparations cancel; on control one this applies
$Q$. Thus only the reflections need controls, including the displayed
sign of $R_0$.

Prepare and retain a uniform choice $j\in\{0,\ldots,M-1\}$ with
$M=4d$, a power of two. Apply $Q$ conditioned on $j\geq k$ for
$k=1,\ldots,M-1$, using~\eqref{eq:adjoint-controlled-grover}.
The choice and comparison registers remain outside $R_0$.
For $q=\sin^2\theta>0$, success is
\[
 \frac1M\sum_{j=0}^{M-1}\sin^2((2j+1)\theta)
 =\frac12-\frac{\sin(4M\theta)}{4M\sin(2\theta)}.
\]
On yes-instances $q\in[1/(32d),1/2]$, so
$\sin(2\theta)\geq1/(4\sqrt d)$ and the result is at least $1/4$.
On no-instances every iterate has zero good amplitude.
Four independent coherent copies, followed by a reversible OR of
their good predicates, accept with probability at least
$1-(3/4)^4=175/256>2/3$ on yes-instances and zero on no-instances.
Only this final OR qubit is measured.

Each register has $O(\log n)$ qubits, and the number of complete
copies is constant. Circuit lengths and loop bounds are polynomial.
The lemma supplies logspace-uniform right-factor circuits; the
remaining gates, reversed circuits and loops also have
logspace-uniform descriptions. The simulator is absolutely halting.
\end{proof}

Steps 3--5 apply to any exact left dilations and approximate right
circuits satisfying~\eqref{eq:rql-actual-step}, provided their actual
inverses and the Clifford+$T$ gates are available exactly. They use
neither controlled versions nor complex conjugates of the additional
source gates. The next subsection reuses these steps after constructing
the required left and right circuits for more general gate sets.

%% file: figure-perfect-soundness.tex
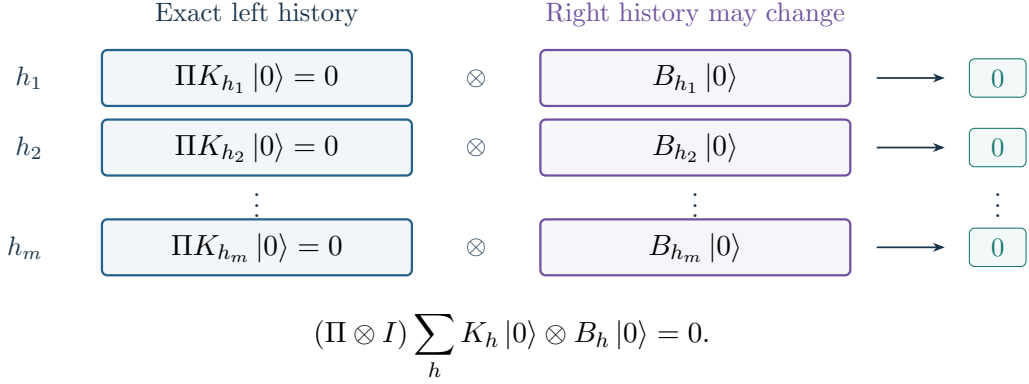
\begin{figure}[tbp]
  \centering

\begin{tikzpicture}[x=1cm,y=.80cm]
    \node[circuit label] at (3.1,1.05) {Exact left history};
    \node[circuit label,text=circuitAux] at (8.9,1.05) {Right history may change};
    \foreach \y/\h in {0/1,-1.15/2,-2.8/m}{
      \node[circuit label,anchor=east] at (.4,\y) {$h_{\h}$};
      \node[circuit gate,minimum width=4.1cm,minimum height=.74cm]
        at (3.1,\y) {$\Pi K_{h_{\h}}\ket0=0$};
      \node[circuit label] at (6,\y) {$\otimes$};
      \node[circuit auxiliary,minimum width=4.1cm,minimum height=.74cm]
        at (8.9,\y) {$B_{h_{\h}}\ket0$};
      \draw[-{Stealth[length=4.5pt]},circuit wire]
        (11.3,\y)--(12.2,\y);
      \node[circuit tag,minimum width=.75cm] at (12.9,\y) {$0$};
    }
    \foreach \x in {3.1,8.9,12.9}{\node[circuit label] at (\x,-1.97) {$\vdots$};}
  \end{tikzpicture}
  \[
    (\Pi\otimes I)\sum_h K_h\ket0\otimes B_h\ket0=0.
  \]
\caption{Perfect soundness holds history by history.
The histories are $h_1,\ldots,h_m$, with
$B_h=B_{T,a_T}\cdots B_{1,a_1}$.
Zero original acceptance forces every accepting left vector to vanish.
Changing its right factor therefore leaves the whole tensor-product term
zero. The final scratch and counter tests select the displayed sum.}
  \label{fig:perfect-soundness}
\end{figure}

%% file: result-adjoints.tex
\subsection{Exact inverses suffice}
\label{sec:rql-adjoints}

We now extend Theorem~\ref{thm:rql-standard} to gate sets whose
entries can be computed within the bounds below and that implement
Clifford+$T$ and each gate's inverse exactly. An exact implementation
is a fixed finite circuit over the gate set, with a constant number of
extra qubits initialized and returned to zero on every input; the circuit
must implement the specified matrix exactly, including its phase.

\begin{theorem}[Measurement elimination with exact inverses]
\label{thm:rql-adjoints}
Let $\mathcal G$ be a fixed finite set of one- and two-qubit unitaries.
Suppose that:
\begin{enumerate}
  \item every entry of every $U\in\mathcal G$ can be approximated to
  error $2^{-b}$ in deterministic space $O(b)$ and time $2^{O(b)}$;
  \item circuits over $\mathcal G$ implement $H,T,\mathrm{CNOT}$ exactly;
  \item for every $U\in\mathcal G$, a circuit over $\mathcal G$
  implements its inverse $U^\dagger$ exactly.
\end{enumerate}
Then $\RQL_{\mathcal G}=\RQUL_{\mathcal G}$.
Neither exact conjugation nor exact controlled versions of the additional
gates are required.
\end{theorem}

\begin{proof}
Put $\Gamma=\{H,T,\mathrm{CNOT}\}$. We construct exact left
circuits $D_t$ over $\mathcal G$ and approximate right circuits $Q_t$
over $\Gamma$. Steps 3--5 of Theorem~\ref{thm:rql-standard} then give
perfect soundness, coherent composition and amplification. The only
new issues are implementing adaptive source steps without controlled
source gates and obtaining the right circuits without exact conjugation.
\paragraph{Implementing the original channel without controlled source gates.}
Pad the quantum tape to at least two cells, and copy the old classical
configuration $c$ into the environment.
For each gate symbol $U\in\mathcal G$, allocate a separate dummy
register of the same arity as $U$, initially zero. There is only one
dummy per gate symbol, not one per classical configuration, so their
total width is constant. Reversible Boolean routing, controlled by the
copied configuration, moves the target qubit or ordered pair of target
qubits to a fixed port. On configurations not requesting $U$, choose
any fixed valid target port. Let $S$ be the swap of this port with its
dummy, and put $S_0=I$, $S_1=S$. If $b(c)=1$ precisely when $U$ is
inactive, apply
\[
 S_{b(c)}(U\otimes I)S_{b(c)}.
\]
Here $U$ is applied unconditionally to the port. For $b(c)=0$ this is
$U\otimes I$; for $b(c)=1$ it is $I\otimes U$.
Undo the routing and retain the dummy in the environment. Thus the
active branch applies $U$ to the data, and every inactive branch leaves
the data unchanged while placing the normalized state $U\ket0$ in
the dummy. The controlled swaps and routing permutations are exact
reversible Boolean circuits over $\Gamma$.

Perform this operation for each of the finitely many gate symbols.
For a measurement step, copy the measured basis bit into an environment
register. Reversible updates compute the new classical configuration
while retaining the old configuration and measurement outcome in the
environment. For example, the working configuration can be cleared by
XOR with its saved copy and then filled with the new configuration.
Polynomial enumeration of configurations and wire positions supplies
logspace-uniform routing and input access. Unused configurations receive
any fixed valid extension of these reversible operations.

Tracing out the environment gives exactly the source transition on
states whose classical configuration is diagonal, which includes every
reachable source state. Copying $c$ deliberately dephases different
classical configurations: we have constructed the required channel,
not a coherent controlled-$U$ gate. Use this chosen channel extension
when defining $\Phi_t$ and its reference states. Its exact dilation
$D_t$ acts on $O(\log n)$ qubits and uses no controlled source gate.
All registers used to retain discarded information, including the
dummies, are included in the Kraus label.
The source acceptance probability $p$ is unchanged.

\paragraph{Right Kraus factors over the fixed standard gate set.}
We claim that Lemma~\ref{lem:rql-auxiliary} still supplies $Q_t$, now
as a circuit over $\Gamma$, for the original chosen dilations
$D_t$ and their original Kraus labels. To prove this, introduce a
\emph{virtual} finite constant-arity gate set containing $\Gamma$,
the gates
\[
 U,\quad U^\dagger,\quad\overline U,\quad U^{\mathsf T}
 \qquad(U\in\mathcal G),
\]
and their singly controlled versions. This virtual gate set has the
exact inverse and conjugation operations used in
Appendix~\ref{app:aux-synthesis}. Further controls reduce to computing
an AND into scratch and using a singly controlled gate. Controls of
the standard gates have exact circuits over $\Gamma$ by
step 1 of Theorem~\ref{thm:rql-standard}. The appendix construction
therefore produces a circuit over this virtual gate set for the right
Kraus factors, with labelled-block error
at most $\varepsilon/2$ and length at most a polynomial $L(n)$.
This is an intermediate circuit description, not an enlargement of
the source machine's gate set.

The entries of each gate $B$ in this virtual set can also be computed
within the stated time and space bounds. To approximate
it with its phase intact, specify the constant-dimensional matrix
\[
 \mathcal E(B)=\ket0\bra0\otimes B+
               \ket1\bra1\otimes B^\dagger,
 \qquad \det\mathcal E(B)=1.
\]
This is a target matrix for synthesis; its definition assumes no
physical controlled-$B$ gate. Space-efficient Solovay--Kitaev synthesis
for fixed dimension~\cite[Theorem~4.5]{MelkebeekWatson12} approximates
$\mathcal E(B)$ in operator norm over $\Gamma$, including its phase.
For completeness, the determinant-one subgroup of the group generated
by fixed-width Clifford+$T$ circuits has finite index and is finitely
generated. Projective universality~\cite{BoykinEtAl99} and connectedness
imply that its closure is the full special unitary group. A fixed finite
inverse-closed generating set of gate sequences therefore meets that theorem's
hypotheses: compactness supplies its fixed constant-accuracy net.
This extends the one-qubit argument in
Appendix~\ref{app:aux-synthesis}.
Only constant dimensions occur here.

Use error $\varepsilon/(2L(n))$ per gate in the virtual set and one shared
extra qubit for the embedding $\mathcal E(B)$, initially zero. The ideal
embedded gates preserve this qubit's zero state and act on that subspace
as the virtual circuit. Their approximations may leave this qubit nonzero.
We retain it coherently and include it among the registers required to
be zero in the selected block. A telescoping bound on the full circuits
bounds their difference
by $\varepsilon/2$, even with this reuse. Compression to the labelled
block cannot increase the error. Thus the resulting $\Gamma$
circuit $Q_t$ has error at most $\varepsilon$ relative to
$\overline W_t$, without changing the target reference states or
assuming exact conjugation over $\mathcal G$.
The required precision is $O(\log n)$ bits; entry computation and
synthesis take polynomial time and logarithmic space to generate the circuit.
The embedding uses only this one extra qubit.

Substitute the exact implementations of the $\Gamma$ gates only after
compiling the circuit for the right Kraus factors.
 Keep the original $D_t$ exact on the first doubled
register. The full step unitary $V_t$ is $D_t$ followed by the actual
$Q_t^\dagger$ on the second register and shared labels, so its clean
block remains exactly
\[
 \widehat C_t=\sum_a K_{t,a}\otimes H_{t,a}^\dagger.
\]
With $\varepsilon=1/(12Td)$ this is precisely the setting of
steps 3--5 of Theorem~\ref{thm:rql-standard}. The history argument
preserves zero accepting amplitude on no-instances, and the product
has squared accepting norm at least $1/(16d)$ on yes-instances.
The unconditional counter construction~\eqref{eq:adjoint-counter-step}
composes these blocks without controlling a source gate. Its inverse
is available by assumption (3). The identity
\eqref{eq:adjoint-controlled-grover} also avoids controlling the
preparation circuit: only exact Boolean reflections need controls.
The same amplification therefore gives acceptance at least
$175/256>2/3$ on yes-instances and zero on no-instances.

The fixed number of dummy registers, shared helper and reused scratch
require $O(\log n)$ qubits. Routing, entry approximation and synthesis
have polynomial running time and logarithmic generation space.
Thus the simulator satisfies the resource and absolute-halting
requirements of $\RQUL_{\mathcal G}$. The reverse inclusion follows
from the definitions.
\end{proof}

\begin{corollary}[Adding gates together with their inverses]
\label{cor:rql-effective-adjoints}
For any fixed finite family of one- and two-qubit unitaries
$U_1,\ldots,U_k$ whose entries satisfy condition~1 of
Theorem~\ref{thm:rql-adjoints}, the gate set
\[
 \mathcal G=\Gamma\cup
 \{U_j,U_j^\dagger:1\le j\le k\}
\]
satisfies $\RQL_{\mathcal G}=\RQUL_{\mathcal G}$.
The entries may be transcendental.
\end{corollary}
\begin{proof}
The inverses of the added gates are already present, while $H$ and
CNOT are self-inverse and $T^\dagger=T^7$.
\end{proof}

Inverses need not be separate primitives. For example,
$\Gamma\cup\{R_\theta\}$ also satisfies the theorem when
\[
 R_\theta=\begin{pmatrix}\cos\theta&-\sin\theta\\
                         \sin\theta&\cos\theta\end{pmatrix}
\]
has entries satisfying condition~1: $R_\theta^\dagger=ZR_\theta Z$
and $Z=T^4$.
This includes $\theta=1$, whose entries are transcendental.
The conclusion here is an equality over the same gate set $\mathcal G$.
It does not by itself identify that class with the class over $\Gamma$;
the next subsection establishes such gate-set independence under a
different, algebraic hypothesis.

%% file: result-algebraic.tex
\subsection{Gate-set independence over CM fields}
\label{sec:rql-algebraic}

Put $\Gamma=\{H,T,\mathrm{CNOT}\}$.
A \emph{number field} is a finite extension of $\mathbb Q$.
An \emph{embedding} into $\mathbb C$ is an injective map preserving
addition, multiplication, and rational numbers; it consistently replaces
algebraic numbers by conjugates. A number field is \emph{totally real}
if every such embedding takes values in $\mathbb R$, or equivalently
if every element's minimal polynomial over $\mathbb Q$ has only real
roots. It is \emph{totally imaginary} if no such embedding takes all
its values in $\mathbb R$.

A \emph{CM field} $F$ is a totally imaginary quadratic extension of a
totally real field $K$: it has dimension two as a vector space over $K$.
Equivalently, $F=K(\sqrt{-a})$ for some $a\in K$ positive under every
embedding of $K$; for example, $\mathbb Q(\sqrt2,i)$ has
$K=\mathbb Q(\sqrt2)$ and $a=1$~\cite[Proposition~1.4]{Milne2020CM}.
Its complex conjugation satisfies
\begin{equation}\label{eq:cm-conjugation}
 \sigma(\overline a)=\overline{\sigma(a)}
 \quad\text{for every embedding }\sigma:F\hookrightarrow\mathbb C
 \text{ and every }a\in F.
\end{equation}
The same identity holds for totally real fields. CM fields are also closed
under finite composita: the smallest field containing finitely many CM
fields is again CM. See
\cite[Proposition~1.4 and Corollary~1.5]{Milne2020CM}.

\begin{theorem}[Algebraic gate-set independence]
\label{thm:rql-cm}
Let $\mathcal G$ be a fixed finite gate set in the local machine model
of the preliminaries. Suppose:
\begin{enumerate}
 \item there is a fixed CM field $F\subseteq\mathbb C$ such that, for
 each $U\in\mathcal G$, some fixed scalar $\lambda_U$ of modulus one
 makes every entry of $\lambda_U U$ belong to $F$;
 \item for each $g\in\Gamma$, there is a fixed finite circuit over
 $\mathcal G$ which, with a constant number of extra qubits initialized to zero,
 implements $\mu_g g$ and returns those ancillas to zero on every input,
 for some fixed scalar $|\mu_g|=1$.
\end{enumerate}
Then
\begin{equation}\label{eq:cm-class-equality}
 \RQL_{\mathcal G}=\RQUL_{\mathcal G}=\RQUL_{\Gamma}.
\end{equation}
The additional gates of $\mathcal G$ need not have exact implementations
of their inverses, complex conjugates, or controlled versions.
\end{theorem}

The field and the gate set are fixed independently of input length.
The proof uses two lemmas. The first replaces field-valued gates by an
intermediate finite gate set of rational orthogonal gates of constant
arity; the second compiles those gates into the original one- and
two-qubit gate set $\Gamma$.

\begin{lemma}[Exact rational simulation of a fixed field]
\label{lem:cm-rational-lift}
Let $F$ be a fixed totally real or CM field of degree $m$.
A polynomial-time, logarithmic-space computation over a fixed finite
gate set with entries in $F$ can be simulated over a fixed
finite rational orthogonal gate set, using $O(1)$ additional qubits.
If its acceptance probability is $p$, the simulator's probability is
\begin{equation}\label{eq:cm-averaged-acceptance}
 \widetilde p=\frac1m\sum_{\sigma:F\hookrightarrow\mathbb C}p_\sigma,
 \qquad p=0\Longrightarrow\widetilde p=0,
 \qquad \widetilde p\ge\frac p m,
\end{equation}
where $p_\sigma$ is the acceptance probability with $\sigma$ applied
entrywise to the original gates. Computational-basis measurements and
classical control are retained in this intermediate simulation.
\end{lemma}
\begin{proof}
Choose a rational basis $b_1=1,b_2,\ldots,b_m$ of $F$, and write
$c(a)\in\mathbb Q^m$ for the coefficient vector of $a\in F$.
Define the normalized trace matrix
\begin{equation}\label{eq:cm-trace-matrix}
 B_{ij}=\frac1m\operatorname{Tr}_{F/\mathbb Q}(\overline{b_i}b_j)
       =\frac1m\sum_\sigma\overline{\sigma(b_i)}\sigma(b_j).
\end{equation}
It is rational and symmetric, with $B_{11}=1$.
It is positive definite because the embedding matrix
$(\sigma(b_i))_{\sigma,i}$ is nonsingular. In particular,
\begin{equation}\label{eq:cm-trace-norm}
 c(a)^{\mathsf T}B c(a)=\frac1m\sum_\sigma|\sigma(a)|^2.
\end{equation}

There is a fixed rectangular rational matrix $E$ with $E^{\mathsf T}E=B$.
Indeed, rational $LDL^{\mathsf T}$ elimination writes
$B=C^{\mathsf T}DC$ with $C$ rational and $D$ positive rational diagonal.
Each diagonal entry $p/q>0$ is a sum of four rational squares: apply
Lagrange's four-square theorem to the integer $pq$ and divide the four
integers by $q$. Combining these four-square vectors gives a rational
factor of $D$, hence of $B$.
Pad the row dimension of $E$ to a power of two $r<8m$ and put
\[
 L=B^{-1}E^{\mathsf T},\qquad P=EL.
\]
Then $LE=I_m$ and $P=P^{\mathsf T}=P^2$ is the orthogonal projector
onto the image of $E$. All these matrices are fixed constants; no
uniform four-square algorithm for growing inputs is being assumed.

For a $k$-dimensional gate $U$ over $F$, let $\rho_k(U)$ be its rational
regular representation: its $(i,j)$ block is multiplication by $U_{ij}$
on field coefficient vectors. This representation preserves products.
By~\eqref{eq:cm-conjugation}, each $\sigma(U)$ is unitary, so
\[
 \rho_k(U)^{\mathsf T}(I_k\otimes B)\rho_k(U)=I_k\otimes B.
\]
Set $E_k=I_k\otimes E$, $L_k=I_k\otimes L$, $P_k=I_k\otimes P$ and
\begin{equation}\label{eq:cm-lifted-gate}
 \widetilde U=E_k\rho_k(U)L_k+(I-P_k).
\end{equation}
This is rational orthogonal and
$\widetilde U E_k=E_k\rho_k(U)$.
For verification, if $X=E_k\rho_k(U)L_k$, then
$X^{\mathsf T}X=P_k$, while the cross terms between $X$ and
$I-P_k$ vanish. Hence $\widetilde U^{\mathsf T}\widetilde U=I$.

Use one \emph{shared} $r$-dimensional field register for the entire
computation, initialized to $v=Ec(1)$. This is a rational unit vector.
It can be prepared from $e_0$ by the fixed rational Householder gate
\[
 I-\frac{2(e_0-v)(e_0-v)^{\mathsf T}}{\norm{e_0-v}^2},
\]
with the identity used if $v=e_0$.
Replace each source gate by~\eqref{eq:cm-lifted-gate} on its physical
wires and this same field register.
Measurements, resets, and classical control act on the physical
registers as before. Each lifted primitive has constant arity.

For an \emph{unnormalized} source history $h$, write $a_{h,i}\in F$
for its amplitude at physical basis coordinate $i$.
Induction on gates and measurement outcomes shows that the corresponding
simulator coordinate is exactly $Ec(a_{h,i})$.
Thus
\[
 \widetilde p
 =\sum_{h,i\ \mathrm{accepting}}c(a_{h,i})^{\mathsf T}B c(a_{h,i})
 =\frac1m\sum_\sigma\sum_{h,i\ \mathrm{accepting}}
      |\sigma(a_{h,i})|^2.
\]
This proves~\eqref{eq:cm-averaged-acceptance}.
If $p=0$, every accepting amplitude is zero under the original embedding
$F\subseteq\mathbb C$, hence is the zero field element and vanishes
under every embedding. The summand for the original embedding also gives
$\widetilde p\ge p/m$.
The loss of $m$ occurs once for the entire computation.

The coefficient encoding is $\mathbb Q$-linear and preserves the
norm defined by the trace matrix $B$; it is not asserted to be a complex-linear isometry
on arbitrary unknown source states. The known initial basis state and
unnormalized histories are what make this simulation valid.
Since $F$, the basis, and the gate set are fixed, all rational lifted
gates and the preparation gate are finite data that can be hardwired.
The extra field register has $O(1)$ qubits, and gate substitution
preserves polynomial time and logarithmic-space uniformity.
\end{proof}

\begin{lemma}[Implementing a rational gate exactly up to a scalar]
\label{lem:cm-scalar-compiler}
Fix a rational orthogonal gate $U$ of constant arity. For every integer
$b\ge1$, put $\delta=2^{-b}$. There is a unitary circuit over $\Gamma$ whose clean
block is \emph{exactly} $hU$ for a scalar $h$ satisfying
$|h|^2\ge1-\delta$. Its size is polynomial in $b$;
its generator, given $b$, uses $O(b)$ space. The number of extra qubits is constant, depending on $U$.
\end{lemma}
\begin{proof}
Choose a positive integer $q$ with $A=qU$ integral, and a power of two
$N\ge\max\{2q,\sum_{ij}|A_{ij}|\}$. These are constants.
Each matrix unit $\ket i\bra j$ is a clean block of a fixed reversible
basis circuit: flag every input label other than $j$, then apply a basis
permutation sending $j$ to $i$. Signs are exact phase gates; a zero block
is obtained by flipping a clean flag.
Using $|A_{ij}|$ signed copies of the appropriate matrix units and
padding with zero blocks gives $N$ such blocks.
Prepare a uniform selector, apply the selected unitary, and unprepare
the selector. Hadamards and reversible Boolean circuits over $\Gamma$
give a fixed unitary $W$ with clean block
\[
 sU,\qquad s=q/N\in(0,1/2].
\]
All temporary scratch registers must be zero in the selected block.

Choose a fixed integer $\ell\ge1$ such that
$\sin(\pi/(4\ell+2))\le s/2$, and put
\[
 a=\frac{\sin(\pi/(4\ell+2))}{s}\in(0,1/2],\qquad
 R_a=\begin{pmatrix}a&-\sqrt{1-a^2}\\\sqrt{1-a^2}&a\end{pmatrix}.
\]
Phase-sensitive one-qubit synthesis from
Appendix~\ref{app:aux-synthesis}, using
\cite[Theorem~4.5]{MelkebeekWatson12}, supplies a
$\Gamma$ circuit $R$ with
$\norm{R-R_a}\le\epsilon=\delta/(2(2\ell+1))$.
Writing $z=\bra0R\ket0$, the unitary $V=R\otimes W$ has clean
block \emph{exactly} $szU$, including any complex phase of $z$.

Let $P_0$ be the clean-scratch projector tensored with the identity on
the data, and use oblivious amplitude amplification with $V$ and its inverse:
\[
 \mathcal A=-V(I-2P_0)V^\dagger(I-2P_0).
\]
The standard two-dimensional amplitude-amplification calculation
\cite{BrassardHoyerMoscaTapp02}, applied to each singular vector, gives
\begin{equation}\label{eq:cm-exact-scalar}
 \operatorname{blk}(\mathcal A^\ell V)=hU,\qquad
 h=\frac z{|z|}\sin\bigl((2\ell+1)\arcsin(s|z|)\bigr),
\end{equation}
with $h=0$ if $z=0$. All singular values of $szU$ are equal, so the
same scalar multiplies every input direction. The reflections and the
inverse of the implemented circuit are exact over $\Gamma$.

For the ideal rotation $R_a$, the amplified clean block is exactly $U$.
There are $2\ell+1$ calls to $V$ or $V^\dagger$ in the amplified circuit.
Unitary telescoping therefore gives
\[
 |h-1|\le(2\ell+1)\epsilon=\delta/2,
 \qquad |h|^2\ge(1-\delta/2)^2\ge1-\delta.
\]
The error estimate is used only to bound the scalar magnitude;
the proportionality in~\eqref{eq:cm-exact-scalar} holds exactly for the
actual synthesized circuit. Since $s$ and $\ell$ are constants, the
claimed resource bounds follow from the one-qubit synthesis bound.
\end{proof}

\begin{proof}[Proof of Theorem~\ref{thm:rql-cm}]
Take a machine for a promise problem in $\RQL_{\mathcal G}$.
Replacing each primitive $U$ by $\lambda_U U$ changes each complete
classically selected history by a phase and leaves all outcome
probabilities unchanged. Apply Lemma~\ref{lem:cm-rational-lift} to the
resulting field-valued gate set. It gives an intermediate rational-gate
simulation with acceptance zero on no-instances and at least $2/(3m)$
on yes-instances, where $m=[F:\mathbb Q]$.

Let $N_0(n)\ge1$ be a polynomial upper bound on the number of rational
gate applications in one such trial, including preparation of the field
register. Apply Lemma~\ref{lem:cm-scalar-compiler} to each of the finitely
many rational gates with
$b=\lceil\log_2(2N_0(n))\rceil$ and
$\delta=2^{-b}\le1/(2N_0(n))$.
After each compiled gate, test whether all its scratch qubits are zero.
Reject if they are not; otherwise reuse them.
This is an ordinary measured computation over $\Gamma$.
The successful block is exactly proportional to the intended rational
gate on every input, including inputs entangled with the rest of the
workspace. Therefore an accepting history's probability is multiplied
by a product of at most $N_0(n)$ factors, each at least $1-\delta$.
Consequently its total acceptance satisfies
\[
 p_\Gamma\ge(1-\delta)^{N_0(n)}\widetilde p
       \ge\tfrac12\widetilde p,
 \qquad \widetilde p=0\Longrightarrow p_\Gamma=0.
\]
On yes-instances this is at least $1/(3m)$.
Run $6m$ independent trials with reused measured workspace and accept
if any trial accepts. Their success is at least
$1-(1-1/(3m))^{6m}\ge1-e^{-2}>2/3$, while no-instance success is zero.
The field and compiler scratch require only constant extra quantum
space. Their finite gate generators and repetition counters use
$O(\log n)$ classical space, and the running time is polynomial.
This proves $\RQL_{\mathcal G}\subseteq\RQL_\Gamma$ using only the stated gate assumptions.

By Theorem~\ref{thm:rql-standard},
$\RQL_\Gamma=\RQUL_\Gamma$.
Finally take the completed unitary $\Gamma$ circuit, expand every
control and inverse into an elementary $\Gamma$ circuit, and only then
substitute the gate implementations from condition (2).
The product of their phases is a single global phase, so this substitution
preserves acceptance exactly and gives
$\RQUL_\Gamma\subseteq\RQUL_{\mathcal G}$.
Together with the defining reverse inclusion, the chain is
\[
 \RQUL_{\mathcal G}\subseteq\RQL_{\mathcal G}
 \subseteq\RQL_\Gamma=\RQUL_\Gamma
 \subseteq\RQUL_{\mathcal G},
\]
which proves~\eqref{eq:cm-class-equality}.
\end{proof}

\begin{corollary}[Algebraic extensions and alternative standard gates]
\label{cor:cm-examples}
The equality~\eqref{eq:cm-class-equality} holds for $\Gamma$ augmented
by any fixed finite collection of one- and two-qubit unitary gates over totally real or CM
fields. In particular, it holds for finite collections of root-of-unity
phase gates or rational orthogonal gates, without adjoining their
inverses or controls.
It also holds for $\{H,P_M,\mathrm{CNOT}\}$ when $8\mid M$, where
$P_M=\operatorname{diag}(1,e^{2\pi i/M})$, and for
$\{H,R_z(\pi/4),\mathrm{CNOT}\}$.
\end{corollary}
\begin{proof}
Take the compositum of the fields of the extra gates with
$\mathbb Q(e^{i\pi/4})$. It is CM: totally real fields can first be
adjoined to $\mathbb Q(i)$, and finite composita of CM fields are CM.
For the displayed alternative gate sets,
$T=P_M^{M/8}$ and $R_z(\pi/4)=e^{-i\pi/8}T$, respectively.
Thus both hypotheses of Theorem~\ref{thm:rql-cm} hold.
\end{proof}

For an example whose entries cannot all lie in a cyclotomic field
(a field generated by a root of unity), let $t$ be a root of
$x^3-4x+1$ and adjoin
\[
 U(t)=\frac1{1+t^2}
 \begin{pmatrix}1-t^2&-2t\\2t&1-t^2\end{pmatrix}
\]
to $\Gamma$. The cubic is irreducible by the rational-root test and
has positive nonsquare discriminant $229$. Hence $\mathbb Q(t)$ is
totally real and non-Galois, and cannot be a subfield of a cyclotomic
field, whose subfields are abelian Galois over $\mathbb Q$.
Moreover, $t=-U(t)_{12}/(1+U(t)_{11})$, so the gate entries themselves
cannot all lie in a cyclotomic field.
Nevertheless Corollary~\ref{cor:cm-examples} applies.

\begin{remark}[Scope of the field hypothesis]
Theorem~\ref{thm:rql-cm} is an equality for fixed gate sets, not a
uniform algorithm over input-dependent number fields.
Its use of primitive global phases is compatible with classical control;
relative phases inside a controlled gate supplied as a primitive must
still satisfy the field hypothesis for that full matrix.
For arbitrary algebraic fields, conjugate gates need not be unitary and
the positive trace form used above need not be available. Thus this proof
does not establish the same conclusion for every algebraic gate set.
\end{remark}

%% file: result-qmal.tex
\section{Perfect completeness with quantum witnesses}
\label{sec:qmal-one}
For two-sided bounded error, the known equality
$\mathsf{QMAL}=\mathsf{Q_UMAL}=\BQL$ already eliminates both
intermediate measurements and logarithmic-size witnesses~\cite[Theorem~1]{FeffermanRemscrim21}.
We ask whether measurements can be eliminated while preserving perfect
completeness: some witness must still be accepted with probability
exactly one on each yes-instance.
Doubling provides such a construction under the gate assumptions below. Here we propagate the final
rejection observable backwards and ask Merlin to supply a state encoding
its history. Checking the history avoids preparing reference states or
implementing matrix inverse square roots.
We first state the gate assumptions, illustrate the construction for
Clifford+$T$, and then give the general construction and its soundness proof.

\subsection{Model and theorem}
\label{sec:qmal-model}
We use the witness-size convention of Fefferman and
Remscrim~\cite[Definition~8]{FeffermanRemscrim21}: the witness and all
stored workspace together have $O(\log n)$ qubits, including classical
memory represented in the computational basis. The verifier is a
logspace-uniform, polynomial-time computation over a fixed finite gate
set $\mathcal G$, with computational-basis measurements, resets, and
classical control allowed. It is absolutely halting: every branch, on
every input and for every witness, halts within a common polynomial
bound. The classical input is read-only and does not count towards the space bound.
The witness may be any state; no circuit-size or gate-set restriction is
imposed on Merlin's preparation.

Write $\mathsf{QMAL}_{1,\mathcal G}$ for the resulting promise class
with acceptance probability one for some witness on each yes-instance
and at most $1/3$ for every witness on each no-instance.
The class $\mathsf{Q_UMAL}_{1,\mathcal G}$ restricts the verifier to
unitary operations until one final computational-basis measurement.
In this section we also allow primitive gates of any fixed constant
arity. This explicitly extends the one- and two-qubit convention used
earlier; the theorem includes that convention as a special case.
We retain the paper's convention that gate entries can be computed to
error $2^{-b}$ in deterministic space $O(b)$ and time $2^{O(b)}$.
This fixes the computational model; it is not an approximation step in
this proof. The primitive gates are applied exactly, and the
construction below never evaluates their entries. In particular, the
exact circuit construction in Lemma~\ref{lem:qmal-exact-block} needs only
symbolic access to the gate names and wire indices.

\begin{theorem}[Measurement elimination with perfect completeness]
\label{thm:qmal-one}
Let $\mathcal G$ be a fixed finite set of constant-arity unitaries whose
entries satisfy the preceding time and space bounds. Suppose that:
\begin{enumerate}
 \item Fixed circuits over $\mathcal G$ implement $H$, $X$, and Toffoli
 exactly, with a constant number of extra qubits that start and end in
 zero for every input state.
 \item For every $U\in\mathcal G$, a fixed circuit over $\mathcal G$
 implements $e^{i\phi_U}U^\dagger$ for some fixed real $\phi_U$, again
 returning its constant number of extra qubits to zero.
\end{enumerate}
Then
\[
 \boxed{\mathsf{QMAL}_{1,\mathcal G}
       =\mathsf{Q_UMAL}_{1,\mathcal G}.}
\]
No exact controlled, conjugated, or transposed versions of the additional
gates are required, and their entries need not be algebraic.
\end{theorem}

The proof below shows why the phase allowed in condition~2 remains
global even when a circuit chooses between different matrix entries.
The reverse containment is immediate; we prove the forward containment.

\subsection{The \texorpdfstring{Clifford+$T$}{Clifford+T} case: checking a rejection history}
\label{sec:qmal-standard}
For $\Gamma=\{H,T,\mathrm{CNOT}\}$, the theorem gives
\[
 \mathsf{QMAL}_{1,\Gamma}=\mathsf{Q_UMAL}_{1,\Gamma}.
\]
The gate assumptions hold exactly: $T^\dagger=T^7$, $H$ and CNOT are
self-inverse, $X=HT^4H$, and Toffoli has an exact Clifford+$T$ circuit
(Section~\ref{sec:rql-standard}). Here is how the unitary verifier works.

\paragraph{1. Describe rejection.}
Let $R$ be the original verifier's rejection operator on its witness
space: a normalized witness $v$ is rejected with probability
$\langle v,Rv\rangle$. Since $R$ is positive semidefinite, perfect
completeness means that some nonzero $v$ satisfies $Rv=0$.

\paragraph{2. Ask Merlin for a history.}
Work backwards from the final rejection projector. Write $B_t$ for
the observable after $t$ backward steps, and $b_t=\vecf(B_t)$.
For $\tau$ elementary channels, Merlin supplies a normalized version of
\[
 h(v)=(v,b_0\otimes v,b_1\otimes v,\ldots,b_\tau\otimes v).
\]
The parentheses denote a direct sum: these are amplitudes in distinct
labelled parts of one state, not separately supplied copies of $v$.
The history has polynomial dimension, so it fits in $O(\log n)$ qubits.
The verifier does not prepare it. Section~\ref{sec:qmal-history} constructs
a matrix $F$ that checks the final rejection projector, every backward
propagation equation, and zero rejection on $v$. In particular,
$Fh(v)=(0,\ldots,0,Rv)$.

\paragraph{3. Test all the equations exactly.}
Pad $F$ by an identity matrix to a power-of-two dimension $N$, obtaining
$A=F\oplus I$. Lemma~\ref{lem:qmal-exact-block} constructs an exact
Clifford+$T$ unitary $V$ whose matrix block between zero-workspace states
is $A/N$. The preliminary test rejects exactly when the verifier workspace
is all zero after $V$. For a normalized certificate $z$, this probability is
\[
 p_{\rm rej}(z)=\frac{\|Az\|^2}{N^2}.
\]
The circuit applies the full unitary $V$ and retains every component;
it does not postselect. An honest history, extended by zero in the
padding coordinates, is never rejected. On a no-instance,
Lemma~\ref{lem:qmal-history-gap} gives an inverse-polynomial lower bound
on this rejection probability for every certificate, including
certificates that do not describe a valid history.

\paragraph{4. Amplify using the same witness.}
Alternate coherent tests of acceptance and initialized workspace,
using $V$ and its actual inverse. A logarithmic-size register records
the first failed test. Section~\ref{sec:qmal-amplification} shows that
polynomially many tests reduce soundness to at most $1/3$, while an
honest certificate passes every test with certainty. Only the final
failure indicator is measured. This uses one supplied certificate and
no approximation of the exact rejection test.

\paragraph{A one-qubit example.}
Apply $H$ to the witness, measure in the computational basis and discard
the outcome, then measure again and accept outcome zero. The backward
observables are
\[
 B_0=B_1=\proj1,\qquad B_2=H\proj1 H=\proj{-}.
\]
Thus $v=\ket+$ is accepted with certainty, and its normalized history
is $h(\ket+)/2$. The exact history test never rejects this certificate.
The first measurement is redundant in this small example; it illustrates
the encoding of an honest history. The general no-instance bound is
proved below. Unlike the RQL simulation, this construction needs no
reference-state preparation or reweighting: Merlin supplies the history,
and the verifier checks it exactly.

\subsection{Encoding the rejection observables in a history}
\label{sec:qmal-history}
We now give the construction for any gate set $\mathcal G$ satisfying
Theorem~\ref{thm:qmal-one}.
We describe rejection by an operator at each stage of the computation.
Starting with the final rejection projector, we propagate this operator
backwards through the verifier. Its expectation at any stage equals the
probability that the rest of the original computation rejects.
Merlin supplies one state encoding a candidate witness together with
the vectorized rejection operators at all stages.
The verifier checks equations fixing the final rejection projector,
relating consecutive operators, and requiring zero rejection on the
candidate witness. An honest certificate satisfies all equations exactly.
On a no-instance, the bound proved below forces every normalized
certificate to violate the combined equations by an inverse-polynomial amount.

\paragraph{Elementary channels and adaptive control.}
Fix a source verifier. It can be represented by a logspace-uniform
sequence $\Phi_1,\ldots,\Phi_\tau$ of polynomially many elementary
channels on $s=O(\log n)$ qubits, each either a primitive unitary gate,
single-qubit dephasing, or a single-qubit reset. This representation
preserves the acceptance probability on every witness.
Here is why it needs no controlled-$U$ assumption. Store classical
configurations in basis registers, copy a measured bit to the appropriate
classical register after dephasing it, and implement deterministic updates
using reversible Boolean gates followed by resets of discarded registers.
For a classically conditional $U$, use a zero dummy register of its arity.
On the inactive classical branch, swap the target with the dummy;
apply $U$ unconditionally; undo the swap; then reset the dummy.
The active branch applies $U$ to the target, and the inactive branch
leaves the target unchanged. Classical configurations are diagonal, so
discarding the dummy gives precisely the desired source channel.

Equality tests, wire routing, and Boolean updates use $X$ and Toffoli
(a fixed helper in state $\ket1$ also supplies CNOT).
Enumerating the polynomially many configurations gives logspace-uniform
circuits of polynomial length and logarithmic width; an update can retain
the old configuration until the new one has been computed and then reset
the old register. Expand the fixed Boolean-gate circuits into primitives
of $\mathcal G$, and include their helper qubits in the total register.
Absolute halting permits padding the source evolution to a fixed length
with an absorbing classical output. The final rejection projector
$\Pi_{\rm rej}$ is diagonal and its diagonal entries are logspace
computable. Thus every chosen elementary map is a channel on the full
register, even on states that are not reachable in the source computation.


Let $D=2^s$ be the total register dimension after this conversion,
$d\leq D$ the original witness dimension, and
$E:\mathbb C^d\longrightarrow\mathbb C^D$ the basis embedding
$E\ket v=\ket v\ket0$. Define
\[
 B_0=\Pi_{\rm rej},\qquad
 B_t=\Phi_{\tau-t+1}^*(B_{t-1})\quad(1\leq t\leq\tau),
 \qquad R=E^\dagger B_\tau E.
\]
Here $\Phi^*$ is the Hermitian conjugate with respect to the trace inner
product. It propagates observables backwards, with
$\Tr(\Phi^*(B)\rho)=\Tr(B\Phi(\rho))$.
For every channel, this backward map is unital and completely positive, so
$0\preceq B_t\preceq I_D$ and $0\preceq R\preceq I_d$.
The original rejection probability on $v$ is $\langle v,Rv\rangle$.
Consequently the promise gives
\begin{equation}
 x\in L_{\rm yes}\ \Longrightarrow\ \ker R\ne\{0\},
 \qquad
 x\in L_{\rm no}\ \Longrightarrow\ R\succeq\delta I_d,
 \label{eq:qmal-rejection-promise}
\end{equation}
where one may take $\delta=2/3$. The proof also works for any known
inverse-polynomial rejection lower bound $0<\delta\leq1$.
Positivity is essential to the first implication:
$\langle v,Rv\rangle=0$ holds exactly when $Rv=0$.

Use row vectorization $\vecf(B)=\sum_{a,b}B_{ab}\ket{a,b}$, and set
$b_t=\vecf(B_t)$. If $M_t$ is the matrix of
$\Phi_{\tau-t+1}^*$ in this representation, then
\[
 b_t=M_tb_{t-1},\qquad
 M_t=\sum_j K_j^\dagger\otimes K_j^{\mathsf T}
 \quad\text{when }\Phi_{\tau-t+1}(X)=\sum_jK_jXK_j^\dagger.
\]
Unless stated otherwise, matrix norms below are Euclidean operator
norms; $\|X\|_F$ is the Frobenius norm and $\|X\|_\infty$ the operator
norm of an observable. A unital completely positive map $\Psi^*$ is
contractive in the latter norm: the Schwarz inequality gives
$\Psi^*(X)^\dagger\Psi^*(X)\preceq\Psi^*(X^\dagger X)
\preceq\|X\|_\infty^2 I$. Therefore
\[
 \|\Psi^*(X)\|_F\leq\sqrt D\|X\|_\infty\leq\sqrt D\|X\|_F.
\]
Every consecutive product of the $M_t$ represents such a map. Hence
\begin{equation}
 \|M_j\cdots M_i\|\leq\sqrt D\quad(1\leq i\leq j\leq\tau),
 \qquad \|b_t\|\leq\sqrt D.
 \label{eq:qmal-product-bound}
\end{equation}
The bound on whole products, rather than a separate normalization at
each step, prevents an exponential loss in time.

Define a linear map
$C:\mathbb C^{D^2}\otimes\mathbb C^d\longrightarrow\mathbb C^d$ by
\[
 C=\sum_{i,j=0}^{d-1}\ket i\bra{Ei}\otimes\bra{Ej}\otimes\bra j.
\]
Then $C(\vecf(B)\otimes v)=E^\dagger BEv$,
$CC^\dagger=dI_d$, and $\|C\|=\sqrt d$.
In particular, this map is not a contraction when $d>1$.
On the direct sum
\[
 \mathcal H=\mathbb C^d\oplus
       \bigoplus_{t=0}^{\tau}(\mathbb C^{D^2}\otimes\mathbb C^d),
 \qquad N_0=d+(\tau+1)D^2d,
\]
write a vector as $Z=(v,z_0,\ldots,z_\tau)$ and define a square matrix
$F$ by its output blocks:
\begin{equation}
 FZ=\bigl(z_0-b_0\otimes v,
          z_1-(M_1\otimes I_d)z_0,\ldots,
          z_\tau-(M_\tau\otimes I_d)z_{\tau-1},Cz_\tau\bigr).
 \label{eq:qmal-constraints}
\end{equation}
Domain and codomain have the same dimension, although the $d$-dimensional
block comes first in the domain and last in the codomain.
Set $h(v)=(v,b_0\otimes v,\ldots,b_\tau\otimes v)$.
Then $Fh(v)=(0,\ldots,0,Rv)$, and the first $\tau+1$ zero constraints
force every kernel vector of $F$ to have this form. Thus
\begin{equation}
 \ker F\ne\{0\}\quad\Longleftrightarrow\quad\ker R\ne\{0\}.
 \label{eq:qmal-kernel}
\end{equation}
Although $h(v)$ encodes a polynomial number of coefficients, its
normalization is a state on only $\lceil\log_2N_0\rceil=O(\log n)$
qubits. Merlin supplies this state; the verifier need not prepare it.
Dishonest witnesses need not factor into history and witness registers.
The following estimate applies to every vector $Z$.

Figure~\ref{fig:qmal-history} aligns the backward observables with
the original computation and summarizes the equations checked on the
certificate.
\input{figure-qmal-history}

\begin{lemma}[History constraint gap]
\label{lem:qmal-history-gap}
On a no-instance, the smallest singular value of $F$ satisfies
\[
 \sigma_{\min}(F)\geq\mu,
 \qquad \mu=\frac{\delta}{3D\sqrt d\,(\tau+1)}.
\]
\end{lemma}
\begin{proof}
Put $q=\tau+1$ and write the output in
\eqref{eq:qmal-constraints} as $(e_0,\ldots,e_\tau,f)$.
Let $\eta=\|FZ\|$ and $w_t=z_t-b_t\otimes v$.
The recurrence for $w_t$ and \eqref{eq:qmal-product-bound}, also using
the identity for an empty product, give
\[
 \|w_\tau\|\leq\sqrt{Dq}\|e\|,\qquad
 \left(\sum_{t=0}^{\tau}\|w_t\|^2\right)^{1/2}
       \leq\sqrt D\,q\|e\|,
 \quad \|e\|^2=\sum_t\|e_t\|^2.
\]
Since $f=Rv+Cw_\tau$ and $R\succeq\delta I$,
\[
 \delta\|v\|\leq\|f\|+\sqrt{dDq}\|e\|
                 \leq\sqrt{1+dDq}\,\eta.
\]
Also $\|h(v)\|\leq\sqrt{1+Dq}\|v\|$. Combining the bounds yields
\[
 \|Z\|\leq
 \left(\frac{\sqrt{(1+Dq)(1+dDq)}}{\delta}+\sqrt D\,q\right)\eta
 \leq\frac{3D\sqrt d\,q}{\delta}\eta.
\]
The last inequality uses $D,d,q\geq1$ and $\delta\leq1$.
\end{proof}

\subsection{An exact test using the original gates}
\label{sec:qmal-coefficients}
Pad by the identity to dimension $N=2^\ell$, the least power of two
at least $N_0$, and put $A=F\oplus I_{N-N_0}$.
Identity padding preserves the kernel and the lower bound $\mu\leq1$;
zero padding would introduce spurious accepting witnesses.
\begin{lemma}[Exact matrix test from paired gate entries]
\label{lem:qmal-exact-block}
Assume the two exact gate-implementation conditions in
Theorem~\ref{thm:qmal-one}. Let $N=2^\ell=\poly(n)$, $N\geq2$,
and suppose a deterministic logspace procedure identifies each entry
of an $N$-by-$N$ matrix $A$ as $0$, $1$, $-1$, or
\[
 \pm\overline{U_{c,a}}U_{d,b}
\]
for a named gate $U\in\mathcal G$ and local basis indices $a,b,c,d$.
Then there is a logspace-uniform circuit $V$ over $\mathcal G$, using
$O(\log N)$ additional qubits and $O(N^2\log N)$ gates, such that
\begin{equation}
 (I_N\otimes\bra{0^a})V(I_N\otimes\ket{0^a})
       =e^{i\theta}\frac{A}{N}.
 \label{eq:qmal-block}
\end{equation}
The phase $\theta$ may depend on the input circuit but is common to
all entries. The constants in the resource bounds depend only on the
fixed gate set and its fixed implementation circuits. No controlled
version of $U$, no entrywise conjugate gate, and no numerical evaluation
of a matrix entry are required.
\end{lemma}
The matrix $F\oplus I$ satisfies the symbolic-entry hypothesis, as we
now check before proving the circuit assertion.

\paragraph{Which coefficients occur?}
The blocks in \eqref{eq:qmal-constraints} occupy disjoint matrix positions.
For a primitive unitary $U$ the backwards step is
$M_U=U^\dagger\otimes U^{\mathsf T}$, with local entries
\begin{equation}
 (M_U)_{(a,b),(c,d)}=\overline{U_{c,a}}U_{d,b}
     =\langle a,d|(U^\dagger\otimes U)|c,b\rangle.
 \label{eq:qmal-paired-entry}
\end{equation}
Untouched wires contribute only Kronecker deltas. Dephasing has a
diagonal matrix with entries in $\{0,1\}$ in the vectorized representation; the backward map for qubit reset sends
$B$ to $B_{00}I$ locally and also has only $0$--$1$ entries.
The entries of $b_0$, $C$, and the padding are $0$ or $1$.
Thus every entry of $A$ is either $0$, $1$, $-1$, or a signed coefficient
in \eqref{eq:qmal-paired-entry}. Given row and column indices, a
logspace machine identifies this form and its indices by inspecting
the source circuit. It need not decide whether a gate coefficient is
numerically zero or compute its value.

\begin{proof}[Proof of Lemma~\ref{lem:qmal-exact-block}]
We first select an entry coherently, then combine the entries into the matrix block.
\paragraph{Selecting a coefficient without a controlled source gate.}
For one such coefficient use registers $R_0,S_0$, each of the arity of
$U$, and a bit $q$ specifying whether this coefficient is selected.
Perform the following operations in order:
\begin{enumerate}
 \item If $q=1$, XOR $c$ into $R_0$ and $b$ into $S_0$.
 \item Apply the available inverse of $U$, up to global phase, to $R_0$.
 \item If $q=0$, swap $R_0$ and $S_0$.
 \item Apply $U$ to $S_0$.
 \item If $q=0$, swap $R_0$ and $S_0$ again.
 \item If $q=1$, XOR $a$ into $R_0$ and $d$ into $S_0$.
\end{enumerate}
Only reversible Boolean operations are controlled. For $q=1$, the
amplitude from zero scratch to zero scratch is
$e^{i\phi_U}\overline{U_{c,a}}U_{d,b}$. For $q=0$, writing $S$ for
the swap, the entire scratch operation is
\[
 S(I\otimes U)S(e^{i\phi_U}U^\dagger\otimes I)
       =e^{i\phi_U}I.
\]
Thus the inactive branch is a scalar identity on all coefficient
scratch states, and the active branch has exactly the same scalar.
This remains true with implementation helpers: keep those helpers
separate from $R_0,S_0$, initialize them to zero, and return them to zero
after each fixed gate circuit. Their zero subspace is invariant, even
when the coefficient scratch is entangled with the address registers.

Enumerate all $N^2$ row--column pairs. At each pair compute its address
equality bit $q$, apply the appropriate circuit, then uncompute $q$.
Use shared coefficient scratch of constant width. A zero entry flips
a separate scratch flag on its selected address; $1$ does nothing,
and a minus sign uses $Z=HXH$ on $q$.
At any fixed address exactly one circuit is active. All the other
circuits are scalar identities on the coefficient scratch, including
any nonzero scratch left by the active circuit. Moreover, each paired
gate circuit is executed on every address and contributes its phase
on both branches. Their product is therefore a common phase independent
of the address, not a relative phase between entries.
The resulting address-preserving unitary $W$ satisfies
\[
 \langle0^r|W_{ij}|0^r\rangle=e^{i\theta}A_{ij}.
\]
Each address test uses $O(\log N)$ Boolean gates and $O(\log N)$
additional qubits returned to zero. There are $N^2$ tests, and the
coefficient circuit at each address has constant size. Thus the gate
count is $O(N^2\log N)$. The circuit generator enumerates addresses
and computes their symbolic entry descriptions in logarithmic space;
it never stores or approximates a complex matrix entry.

\paragraph{From entries to a matrix block.}
Figure~\ref{fig:qmal-verifier}(a) shows the circuit described here.
Start with a certificate in a column register $J$ and a zero row
register $I$, both of $\ell$ qubits. Apply $H^{\otimes\ell}$ to $I$,
apply $W$ controlled by the two addresses as just constructed, apply
$H^{\otimes\ell}$ to $J$, and swap $I,J$.
Selecting zero coefficient scratch and zero old-column register gives
\[
 \frac{e^{i\theta}}{N}\sum_{i,j}A_{ij}\ket i\bra j,
\]
which proves \eqref{eq:qmal-block}. Here ``selecting'' specifies a
subspace, not an intermediate measurement or postselection. The
Hadamard layers and final swap add $O(\log N)$ gates.
\end{proof}
Define rejection to be the event that all this workspace is zero.
For a normalized certificate $h$, the rejection probability is exactly
\begin{equation}
 p_{\rm rej}(h)=\frac{\|Ah\|^2}{N^2}.
 \label{eq:qmal-test-probability}
\end{equation}
A reversible Boolean computation can put this zero test in a final
output bit if desired. Equations~\eqref{eq:qmal-kernel} and
\eqref{eq:qmal-test-probability} give perfect completeness and,
on no-instances, rejection probability at least
\[
 \gamma=\frac{\mu^2}{N^2}
        =\frac{\delta^2}{9D^2d(\tau+1)^2N^2}
        \geq\frac1{\poly(n)}.
\]

\subsection{Unitary amplification with one witness}
\label{sec:qmal-amplification}
We give the amplification explicitly to retain the exact gate
conditions and perfect completeness. Space-efficient amplification
for unitary computations with quantum witnesses was developed by
Fefferman et al.~\cite{FeffermanKLMN16}.
Let $P=I_N\otimes\proj{0^a}$ project onto initialized verifier
workspace, let $\Pi_{\rm acc}=I-P$ be the preliminary test's accepting
projector at the output of $V$, and put
$Q=V^\dagger\Pi_{\rm acc}V$.
On the image of $P$, \eqref{eq:qmal-block} gives
\begin{equation}
 PQP\simeq I_N-\frac{A^\dagger A}{N^2}.
 \label{eq:qmal-projectors}
\end{equation}
An initialized kernel witness is fixed by both $P$ and $Q$.
On a no-instance, $0\preceq PQP\preceq(1-\gamma)P$.
Alternating a $Q$ test and a $P$ test for $k$ pairs, accepting only if
every test passes, would therefore have acceptance probability at most
\[
 \|(PQP)^k(h\otimes\ket{0^a})\|^2
       \leq(1-\gamma)^{2k}\leq e^{-2k\gamma}.
\]

These tests can be performed coherently in logarithmic space.
Figure~\ref{fig:qmal-verifier}(b) shows one pair of tests and the final
measurement.
Introduce a failure-label register with labels $0,1,\ldots,2k$,
initially zero; any unused binary labels are left fixed.
At test $j$, transpose labels $0$ and $j$ if that test fails.
For a $Q$ test apply $V$, condition the label transposition on
$\Pi_{\rm acc}$ failing, and apply $V^\dagger$.
For a $P$ test condition it on nonzero verifier workspace.
These are Boolean conditions in the computational basis between the
unconditional calls to $V$ and its inverse. Each predicate is computed
and uncomputed using separate scratch.
Before test $j$, label $j$ has no amplitude: earlier tests touched only
smaller labels. Thus an old failure never returns to label zero, and
the label-zero component is exactly the product of the passing
projectors applied so far. This proves the displayed probability bound
for the fully unitary implementation and for arbitrary, possibly
entangled certificates. Mixed witnesses follow by linearity.

Every gate of $V$ has an inverse over $\mathcal G$ up to a global phase.
Reversing its circuit therefore implements $V^\dagger$ up to a common
phase, which does not affect either projector test.
Extra helpers used to reverse primitive gates are returned to zero
after each inverse circuit. No controlled-$V$ operation is used.
Take $k=\lceil1/\gamma\rceil$ using the explicit rational bound above
(with $\delta=2/3$). This is a logspace-computable polynomial bound.
The final verifier accepts exactly when the failure label is zero,
with perfect completeness and soundness at most $e^{-2}<1/3$.
It uses $O(\log N+\log k)=O(\log n)$ witness and workspace qubits,
polynomially many gates, and a fixed schedule on every input.
This completes the proof of Theorem~\ref{thm:qmal-one}.

\input{circuit-qmal-verifier}

\subsection{Gate-set examples and comparison with reweighted doubling}
\label{sec:qmal-comparison}
Section~\ref{sec:qmal-standard} verifies the conditions for
$\Gamma=\{H,T,\mathrm{CNOT}\}$; see also the exact-synthesis
characterization of Giles and Selinger~\cite{GS13}.
Further examples are:
\begin{enumerate}
 \item For any fixed finite collection $\mathcal U$ of constant-arity
 unitaries whose entries satisfy the stated time and space bounds,
 $\{H,X,\mathrm{Toffoli}\}\cup\mathcal U\cup\mathcal U^\dagger$
 satisfies the theorem. Clifford+$T$ universality is not needed.
 \item For any fixed finite set $\Theta$ of angles such that each
 $e^{i\theta}$ can be computed within those bounds,
 $\Gamma\cup\{P_\theta:\theta\in\Theta\}$ works, where
 $P_\theta=\operatorname{diag}(1,e^{i\theta})$, since
 $XP_\theta X=e^{i\theta}P_\theta^\dagger$.
 The angles need not be rational multiples of $\pi$.
 \item The set $\Gamma\cup\{R_5\}$ works for
 \[
 R_5=\begin{pmatrix}3/5&-4/5\\4/5&3/5\end{pmatrix},
 \qquad ZR_5Z=R_5^\dagger.
 \]
 Its entries are outside $\mathbb Z[1/\sqrt2,i]$, so this specified
 matrix has no exact Clifford+$T$ implementation with extra qubits
 returned to zero~\cite{GS13}. The construction uses $R_5$ directly.
\end{enumerate}
These are equalities within each qualifying gate set, not a claim that
different such gate sets define the same perfectly complete class.

Both this construction and the RQL construction use vectorization and
preserve an exact zero, but they simulate different operators. Section~\ref{sec:doubling-framework} propagates
a density operator forwards, reweights its doubled evolution, and
keeps the left Kraus factors exact while approximating the right Kraus factors.
Here we propagate a rejection observable backwards. Merlin supplies a
history certificate, the verifier tests every history equation
through one exact matrix block, and a bound on products of backward
maps on observables supplies the inverse-polynomial gap.
There are no reference states, square-root or inverse-square-root
approximations, or invocations of $\BQL=\BQUL$ in this proof.
The first-failure label serves the same purpose as the earlier failure
counter, but now enforces a sequence of alternating projector tests.

The exact condition also differs: RQL has perfect soundness on
no-instances, whereas the present theorem preserves the existence of a
perfectly accepted witness on yes-instances. It does not remove that
witness or make completeness perfect when it was originally
less than one. In particular, it does not establish
$\mathsf{QMAL}=\mathsf{QMAL}_1$ or
$\RQL=\mathsf{coQ_UMAL}_1$, the distinct questions discussed by
Fefferman and Remscrim~\cite[Open Problems~1--2]{FeffermanRemscrim21}.

\paragraph{Comparison with perfect completeness for QMA.}
Grewal and Rudolph prove $\mathsf{QMA}=\mathsf{QMA}_1$ over
$\{X,\mathrm{CNOT},\mathrm{Toffoli},H\otimes H\}$ using an exact
history constraint~\cite[Sections~3.1--3.3]{GrewalRudolph26}.
They first arrange that the acceptance operator $E$ has dyadic rational
entries with a known denominator. Merlin supplies an exact rational
value $\Gamma_z=\langle z|(2E)^\ell|z\rangle$ and a history whose
final vector is mapped back to its first vector using $\Gamma_z^{-1}$.
An exact matrix test rejects in proportion to the squared norm of the
constraint residual. This is closely related to our use of a history
certificate and an exact kernel test.

The space requirement separates the two results. Their certificate
includes the numerator and denominator of $\Gamma_z$, with
$O(h\ell)$ bits when $2^hE$ is integral, and their implementation uses
an index register of that order. Even for a logarithmic-size original
witness, $h$ can grow with the polynomial circuit length; the resulting
space bound need not be logarithmic. Our construction instead uses the
original perfect-completeness promise to obtain $\ker R\ne\{0\}$ and
checks its backward history without a supplied rational scalar. Thus
their theorem does not directly establish
$\mathsf{QMAL}=\mathsf{QMAL}_1$, and our theorem only removes
measurements from verifiers that already have perfect completeness.

\subsection{Connection to span programs}
\label{sec:qmal-span-programs}
Span programs, introduced by Karchmer and Wigderson~\cite{KarchmerWigderson93},
decide whether a fixed target vector $t$ lies in the span of vectors
made available by the input bits. In linear-map notation, a fixed map
$L:\mathcal U\to\mathcal V$ and an available subspace $\mathcal U(x)$
accept $x$ when $Lw=t$ for some $w\in\mathcal U(x)$.
Such a $w$ is a \emph{positive witness}; a \emph{negative witness}
is a vector $\omega$ with $\langle\omega,t\rangle=1$ and
$\langle\omega,Lu\rangle=0$ for every $u\in\mathcal U(x)$.
Witness norms underlie the connection to quantum query
complexity~\cite{ReichardtSpalek12,Reichardt09Span}.

Our history test also describes a vector separating a target from a
linear subspace. Writing $F_x$
for the matrix in~\eqref{eq:qmal-constraints} and $e_j$ for the
standard basis of its domain, finite-dimensional linear algebra gives
\[
 \ker F_x=(\operatorname{im}F_x^\dagger)^\perp,
 \qquad
 \ker F_x\ne\{0\}
 \ \Longleftrightarrow\
 \ \exists j:\ e_j\notin\operatorname{im}F_x^\dagger.
\]
Indeed, any nonzero $Z\in\ker F_x$ has a nonzero coordinate and can
be rescaled to separate the corresponding $e_j$ from every column of
$F_x^\dagger$. Thus a perfectly accepted history certificate supplies
the same kind of separating vector as a negative span-program witness.
On a no-instance, Lemma~\ref{lem:qmal-history-gap} instead makes
$F_x^\dagger$ invertible, and every equation $F_x^\dagger w=e_j$
has a solution of norm at most $1/\mu$.

There is also a spectral connection. Reichardt and \v{S}palek use
zero-eigenvalue witnesses in span-program formula
evaluation~\cite{ReichardtSpalek12}; Reichardt's reflection formulation
gives query-optimal algorithms~\cite{reichardt2010reflections}.
Belovs and Reichardt obtain logarithmic-space algorithms for
$st$-connectivity~\cite{BelovsReichardt12}, while Ito and Jeffery develop
approximate witnesses and their relation to spectral
gaps~\cite{ItoJeffery19}.
In our test,~\eqref{eq:qmal-projectors} identifies the perfectly
accepted certificates with the common fixed vectors of $P$ and $Q$;
the singular-value bound for $F_x$ gives the gap used in amplification.

The resource connection is particularly close to Jeffery's bounds
relating quantum space to the logarithm of span-program size, including
approximate span programs for two-sided error~\cite{Jeffery22SpanSpace}.
Cornelissen, Jeffery, Ozols and Piedrafita also give conversions between
quantum algorithms and span programs that preserve time up to
polylogarithmic factors~\cite{CornelisenJefferyOzolsPiedrafita20}.
Here polynomial matrix dimension likewise permits a logarithmic-size
quantum certificate. However, $F_x$ depends on the full input, whereas
a standard span program fixes its map and lets input bits select
available vectors. The displayed identity is therefore a geometric
connection, not by itself a construction of such a span program or an
algorithm without a witness. Our theorem additionally requires Merlin
to supply the certificate and the exact gate implementation in
Section~\ref{sec:qmal-coefficients} to preserve perfect completeness.

%% file: figure-qmal-history.tex
\begin{figure}[tbp]
  \centering
  \begin{tikzpicture}[x=1cm,y=1cm,
    history arrow/.style={-{Stealth[length=4.5pt]},circuit wire}]
    \node[circuit label,anchor=west,font=\small\bfseries] at (-.1,1)
      {The original verifier evolves states to the right};
    \foreach \x/\lab in {1/0,4.3/1,9.1/{\tau-1},12.4/\tau}{
      \node[circuit gate,minimum width=1.35cm,minimum height=.65cm]
        at (\x,0) {$\rho_{\lab}$};
    }
    \draw[history arrow] (1.75,0)--node[circuit label,above] {$\Phi_1$} (3.55,0);
    \draw[history arrow] (9.85,0)--node[circuit label,above] {$\Phi_\tau$} (11.65,0);
    \node[circuit label] at (6.7,0) {$\cdots$};
    \node[circuit label,anchor=west,font=\small\bfseries] at (-.1,-1.05)
      {Rejection observables propagate to the left};
    \foreach \x/\lab in {1/\tau,4.3/{\tau-1},9.1/1,12.4/0}{
      \node[circuit auxiliary,minimum width=1.35cm,minimum height=.65cm]
        at (\x,-2.1) {$B_{\lab}$};
    }
    \draw[history arrow] (3.55,-2.1)--node[circuit label,above] {$\Phi_1^*$} (1.75,-2.1);
    \draw[history arrow] (11.65,-2.1)--node[circuit label,above] {$\Phi_\tau^*$} (9.85,-2.1);
    \node[circuit label] at (6.7,-2.1) {$\cdots$};
    \node[circuit label] at (1.3,-2.85) {$R=E^\dagger B_\tau E$};
    \node[circuit label,text=circuitClean] at (11.75,-2.85)
      {Start with $B_0=\Pi_{\rm rej}$};
    \node[circuit tag,align=center,text width=13.5cm,minimum height=1.1cm]
      at (6.7,-4.2)
      {An honest certificate is the single state $h(v)/\|h(v)\|$\\[2pt]
       $h(v)=(v,b_0\otimes v,\ldots,b_\tau\otimes v),
        \qquad b_t=\vecf(B_t)$};
    \foreach \x in {1.65,6.7,11.75}
      \draw[history arrow] (\x,-4.82)--(\x,-5.35);
    \node[circuit gate,text width=3.85cm,minimum height=1.15cm,font=\small]
      at (1.65,-6.02)
      {\textbf{Fix the endpoint}\\[3pt] $z_0=b_0\otimes v$};
    \node[circuit gate,text width=4.25cm,minimum height=1.15cm,font=\small]
      at (6.7,-6.02)
      {\textbf{Check consecutive steps}\\[3pt]
       $z_t=(M_t\otimes I_d)z_{t-1}$};
    \node[circuit gate,text width=3.85cm,minimum height=1.15cm,font=\small]
      at (11.75,-6.02)
      {\textbf{Require zero rejection}\\[3pt] $Cz_\tau=0$};
  \end{tikzpicture}
  \caption{The backward history for $\mathsf{QMAL}_1$.
    For a candidate witness $v$, the original verifier starts in
    $\rho_0=\ket{Ev}\bra{Ev}$ and applies
    $\rho_j=\Phi_j(\rho_{j-1})$.
    The aligned observable is $B_{\tau-j}$, so
    $\operatorname{Tr}(B_{\tau-j}\rho_j)$ is the same rejection
    probability at every time $j$.
    The tuple $h(v)$ denotes an unnormalized vector in a direct sum;
    its normalized encoding uses $O(\log n)$ qubits, not a separate
    copy of every operator. For a general certificate
    $Z=(v,z_0,\ldots,z_\tau)$, the three sets of equations are precisely
    $FZ=0$. They force $Z=h(v)$ and $Rv=0$.
    An honest witness on a yes-instance satisfies them exactly; the
    gap bound of Lemma~\ref{lem:qmal-history-gap} controls their combined violation for every
    normalized certificate on a no-instance.}
  \label{fig:qmal-history}
\end{figure}

%% file: circuit-qmal-verifier.tex
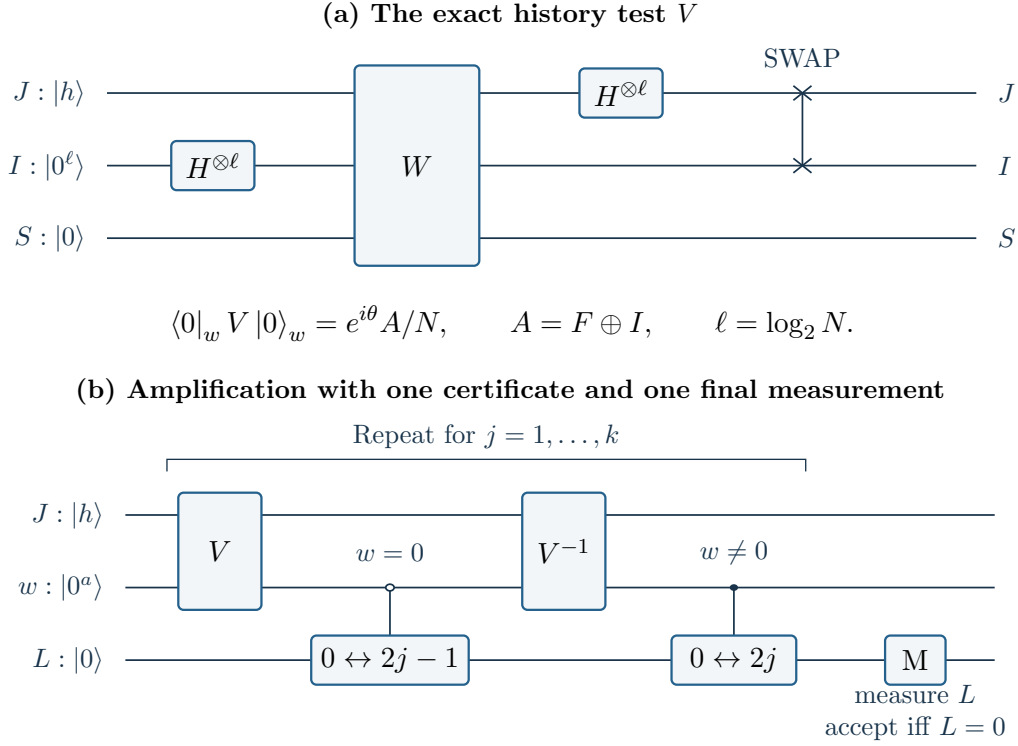
\begin{figure}[tbp]
  \centering
  {\small\textbf{(a) The exact history test $V$}}\par\smallskip
  \begin{tikzpicture}[x=1cm,y=.8cm]
    \foreach \y in {0,1.2,2.4}{\draw[circuit wire] (0,\y)--(11.5,\y);}
    \node[circuit label,anchor=east] at (-.15,2.4) {$J:\ket h$};
    \node[circuit label,anchor=east] at (-.15,1.2) {$I:\ket{0^\ell}$};
    \node[circuit label,anchor=east] at (-.15,0) {$S:\ket0$};
    \node[circuit gate,minimum width=1.1cm,minimum height=.65cm]
      at (1.4,1.2) {$H^{\otimes\ell}$};
    \node[circuit gate,minimum width=1.65cm,minimum height=2.65cm]
      at (4.1,1.2) {$W$};
    \node[circuit gate,minimum width=1.1cm,minimum height=.65cm]
      at (6.8,2.4) {$H^{\otimes\ell}$};
    \draw[circuit wire] (9.2,1.2)--(9.2,2.4);
    \foreach \y in {1.2,2.4}{
      \draw[circuit wire] (9.08,\y-.12)--(9.32,\y+.12);
      \draw[circuit wire] (9.08,\y+.12)--(9.32,\y-.12);
    }
    \node[circuit label] at (9.2,2.98) {SWAP};
    \node[circuit label,anchor=west] at (11.65,2.4) {$J$};
    \node[circuit label,anchor=west] at (11.65,1.2) {$I$};
    \node[circuit label,anchor=west] at (11.65,0) {$S$};
  \end{tikzpicture}
  \[
    \bra0_w V\ket0_w=e^{i\theta}A/N,
    \qquad A=F\oplus I,\qquad \ell=\log_2N.
  \]
  {\small\textbf{(b) Amplification with one certificate and one final measurement}}\par\smallskip
  \begin{tikzpicture}[x=1cm,y=.8cm]
    \foreach \y in {0,1.2,2.4}{\draw[circuit wire] (0,\y)--(11.5,\y);}
    \node[circuit label,anchor=east] at (-.15,2.4) {$J:\ket h$};
    \node[circuit label,anchor=east] at (-.15,1.2) {$w:\ket{0^a}$};
    \node[circuit label,anchor=east] at (-.15,0) {$L:\ket0$};
    \node[circuit gate,minimum width=1.1cm,minimum height=1.55cm]
      at (1.25,1.8) {$V$};
    \draw[circuit wire] (3.5,0)--(3.5,1.2);
    \draw[draw=circuitInk,fill=white,line width=.65pt] (3.5,1.2) circle (.055);
    \node[circuit label] at (3.5,1.78) {$w=0$};
    \node[circuit gate,minimum width=1.85cm,minimum height=.65cm]
      at (3.5,0) {$0\leftrightarrow 2j-1$};
    \node[circuit gate,minimum width=1.1cm,minimum height=1.55cm]
      at (5.8,1.8) {$V^{-1}$};
    \draw[circuit wire] (8.05,0)--(8.05,1.2);
    \fill[circuitInk] (8.05,1.2) circle (.055);
    \node[circuit label] at (8.05,1.78) {$w\ne0$};
    \node[circuit gate,minimum width=1.65cm,minimum height=.65cm]
      at (8.05,0) {$0\leftrightarrow 2j$};
    \draw[draw=circuitInk,line width=.5pt]
      (.55,3.12)--(.55,3.32)--(9,3.32)--(9,3.12);
    \node[circuit label] at (4.775,3.7) {Repeat for $j=1,\ldots,k$};
    \node[circuit gate,minimum width=.8cm,minimum height=.65cm]
      at (10.45,0) {$\mathrm M$};
    \node[circuit label,align=center] at (10.45,-.87)
      {measure $L$\\accept iff $L=0$};
  \end{tikzpicture}
  \caption{The unitary verifier constructed from the original
    $\mathsf{QMAL}_1$ verifier.
    Merlin supplies one normalized history certificate $\ket h$ in $J$.
    In (a), $W$ implements the matrix entries of $A$ exactly up to a common
    phase; $S$ groups coefficient scratch and temporary helpers, and
    $w$ comprises $I$ and $S$. The zero-workspace output of $V$ has
    probability $\|Ah\|^2/N^2$ and is the \emph{rejecting} outcome.
    In (b), each controlled box swaps the displayed labels of $L$ and
    fixes all other labels. The first records failure of the $Q$ test;
    after $V^{-1}$, the second records failure of the $P$ test.
    Conditions concern the whole workspace $w$, and their temporary
    helpers are returned to zero. The inverse of $V$ is implemented up
    to global phase, which does not affect these tests.
    A fresh nonzero label is used at every test, so a recorded failure
    never returns to label zero. An honest certificate in $\ker A$
    passes every test with certainty; all operations before the final
    measurement are unitary.}
  \label{fig:qmal-verifier}
\end{figure}

%% file: conclusion.tex
\section{Conclusion}
Under the gate assumptions of Theorem~\ref{thm:rql-adjoints},
$\RQL_{\mathcal G}=\RQUL_{\mathcal G}$. Keeping the left Kraus factors
exact preserves zero acceptance, while reference-state reweighting
retains an inverse-polynomial accepting probability on yes-instances.
Theorem~\ref{thm:rql-cm} also identifies this class with the standard
Clifford+$T$ class under its CM-field assumptions.

Theorem~\ref{thm:qmal-one} proves
$\mathsf{QMAL}_{1,\mathcal G}=\mathsf{Q_UMAL}_{1,\mathcal G}$
under its stated gate assumptions. Here Merlin supplies a history
certificate, and exact checks preserve perfect completeness.
This construction uses neither reference-state reweighting nor the
two-sided measurement-elimination theorem used by the RQL construction.

Several questions follow from these constructions.

\paragraph{Which gate sets support the simulation?}
Can the exact inverse assumption in Theorem~\ref{thm:rql-adjoints} be
removed while retaining its bounds on computing gate entries?
Can the requirement to implement Clifford+$T$ in
Theorems~\ref{thm:rql-adjoints} and~\ref{thm:rql-cm} be removed?
Can the gate-set independence theorem be extended beyond fixed CM-field gate sets?
Approximate universality does not settle this issue: the standard gate-conversion argument preserves a bounded gap but can destroy an exact zero~\cite{MelkebeekWatson12,LeGallLiuWang23}.

\paragraph{Can both acceptance and rejection be exact?}
Define $\mathsf{EQL}_{\mathcal G}$ and $\mathsf{EQ_UL}_{\mathcal G}$
using the same absolutely halting, polynomial-time, logarithmic-space
models, with a correct answer of probability one on every promised input.
Whether $\mathsf{EQL}_{\mathcal G}=\mathsf{EQ_UL}_{\mathcal G}$ remains
open here. The missing requirement is an exact implementation of the
operations used by an exact simulation.

One can remove a normalization obstacle at the level of ideal operators.
For the pure initial state, take $\tau_0=\rho_0$ and
$\tau_t=\Phi_t(\tau_{t-1})+3I/(dT)$. The same doubled construction,
with no inverse of $\tau_0$, gives
$\vecf(\rho_T\tau_T^{-1/2})$.
Since $\Tr\tau_T=4$, pairing with
$\vecf(\tau_T^{1/2})/2$ after copying the output gives amplitudes
$p_b/2$ for the original output probabilities $p_b$.
For an exact source computation this would give success probability
$1/4$, and one ordinary Grover iteration would suffice. Exact efficient
implementations of the steps and final pairing over a fixed gate set
are not supplied by our approximation arguments. Establishing such an
implementation, possibly with different reweighting factors, is the
open problem; no EQL equality follows from this observation.
\paragraph{What is the remaining role of a quantum witness?}
Theorem~\ref{thm:qmal-one} removes intermediate measurements from
perfectly complete verifiers but retains a quantum witness. The further equality
$\RQL=\mathsf{coQ_UMAL}_1$ posed by Fefferman and Remscrim is not automatic~\cite[Open Problem~1]{FeffermanRemscrim21}.
Can the argument that preserves exact rejection help eliminate witnesses
in this perfect-completeness setting?
Any such comparison must use compatible gate-set conventions for the verifier and the simulation.
The separate question $\mathsf{QMAL}=\mathsf{QMAL}_1$ also remains:
the history construction preserves an existing exact kernel and does
not create one from a merely high-acceptance witness.

\paragraph{Can the reference-state computation be made cheaper?}
The reference states bound the norms of the doubled steps, but
implementing their square roots and inverse square roots can dominate
the cost of the simulation.
Can structure in the source channels make it cheaper to access the
reference states and implement their square roots and inverse square roots?
Can these operations be implemented without using two-sided measurement
elimination?
For superlogarithmic space, the relevant target is an improved simultaneous time--space bound, not just polynomial time in the Hilbert-space dimension.
The time--space trade-offs of Girish and Raz, together with Zhandry's
black-box lower bounds, provide benchmarks for this question~\cite{GirishRaz22,Zhandry24}.

%% file: appendix-doubling.tex
\section{Implementing doubling}
\label{app:implementing-doubling}

This appendix collects the circuit constructions for the reweighted doubled operators defined in the main text.
The ideal clean-block construction is shared by the two simulations; the compiled construction makes explicit which factors remain exact in the one-sided case.

\subsection{Implementation with exact controls and conjugation}
\label{app:controlled-interface}
The following gate assumptions give a direct generalization of
Section~\ref{sec:rql-standard}. They also specify the interface used by
Lemma~\ref{lem:rql-auxiliary}. Section~\ref{sec:rql-adjoints} removes the
requirements for exact controlled gates and complex conjugates.
\input{result-rql}

\clearpage
\subsection{Doubling for \texorpdfstring{$\BQL$}{BQL}}

\label{app:doubling-bql}

Figure~\ref{fig:circuit-Ct} expresses the ideal contraction $C_t$ as a clean block of a unitary circuit.
The unitary extension of $\overline{W_t}$ in this diagram is an abstract operator, so the diagram alone does not assert exact synthesis over a fixed gate set.
For bounded error, this operator may be approximated to the accuracy required by the full simulation.
The compiled form in Appendix~\ref{app:doubling-rql} can also be used for $\BQL$.
Appendix~\ref{sec:doubling-bql}
 proves the acceptance bounds after reducing the original computation's
error. It also explains why implementing the right Kraus factors uses
the known measurement-elimination theorem.

\input{circuit-Ct}

\clearpage
\subsection{Doubling for \texorpdfstring{$\RQL$}{RQL}}
\label{app:doubling-rql}

For the one-sided simulation, Figure~\ref{fig:circuit-Chat-t} replaces the ideal unitary for the right Kraus factors by a compiled circuit $Q_t$ while retaining the exact unitary dilation $D_t$ of the original channel.
Thus the clean block $\widehat C_t$ retains the exact Kraus factors on the first data register, as required by the perfect-soundness argument in the main text.
The gate-set implementation hypotheses stated there still apply.

\input{circuit-Chat-t}

\clearpage
\input{appendix-auxiliary}

%% file: result-rql.tex
We impose the following sufficient conditions on the fixed finite gate set
$\mathcal G$. An \emph{exact implementation} is a fixed finite circuit over
$\mathcal G$, with a constant number of extra qubits initialized and
returned to zero on every input. The circuit must implement the specified
matrix exactly, including its phase.
\begin{enumerate}
  \item \textbf{Computing gate entries.} For every $U\in\mathcal G$, its entries
  can be approximated to additive error $2^{-b}$ in deterministic space
  $O(b)$ and time $2^{O(b)}$. Only $b=O(\log n)$ is used below.
  \item \textbf{Exact standard gates.} There are exact implementations of
  $H$, $T=\operatorname{diag}(1,e^{i\pi/4})$, and $\mathrm{CNOT}$.
  In particular, reversible Boolean gates, computational-basis predicates,
  and their phase reflections can be implemented exactly.
  \item \textbf{Exact inverses and complex conjugates.} For every
  $U\in\mathcal G$, there are exact implementations of $U^\dagger$ and
  $\overline U$, where conjugation is in the computational basis.
  \item \textbf{Exact coherent control.} For every $U\in\mathcal G$,
  there is an exact implementation of
  $\ket0\bra0\otimes I+\ket1\bra1\otimes U$.
  Further controls are obtained by computing their conjunction into clean
  scratch. In particular, this condition applies to the actual finite
  circuits implementing the right Kraus factors, not just their ideal targets.
\end{enumerate}
These fixed circuits are independent of input length. Approximate scalar
rotations are compiled using the exact implementations of Clifford+$T$
and space-efficient
synthesis~\cite{MelkebeekWatson12}; the phase-sensitive convention needed
for controlled rotations is explained in Appendix~\ref{app:aux-synthesis}.
No exact synthesis of $\sigma_t^{1/2}$ or $\sigma_t^{-1/2}$ is assumed.

\begin{theorem}[Same-gate-set measurement elimination]
\label{thm:rql-same-gates}
Under conditions (1)--(4), the polynomial-time, logarithmic-space classes
defined above satisfy
\[
  \RQL_{\mathcal G}=\RQUL_{\mathcal G}.
\]
\end{theorem}

\begin{proof}
Under conditions (1)--(4), the construction in step 1 of
Theorem~\ref{thm:rql-standard} gives exact dilations $D_t$ over
$\mathcal G$. Lemma~\ref{lem:rql-auxiliary} supplies the right-factor
circuits $Q_t$ over the same gate set. Conditions (2)--(4) allow every
inverse, control and Boolean reflection used there to be implemented
exactly. The clean-block identity~\eqref{eq:rql-actual-step} and
steps 3--5 of that proof therefore give an absolutely halting,
polynomial-time unitary simulation in logarithmic space, with zero
acceptance on no-instances and acceptance greater than $2/3$ on
yes-instances. The reverse inclusion follows from the definitions.
\end{proof}
\begin{corollary}[Examples]\label{cor:rql-gate-examples}
The theorem applies to $\Gamma=\{H,T,\mathrm{CNOT}\}$ and to
\[
 \mathcal G_\theta=\Gamma\cup\{P_\theta,P_{-\theta}\},
 \qquad P_\theta=\operatorname{diag}(1,e^{i\theta}),
\]
whenever the added entries satisfy condition (1), including $\theta=1$.
\end{corollary}
\begin{proof}
The exact inverse, conjugation and controlled circuits for $\Gamma$
were given in step 1 of Theorem~\ref{thm:rql-standard}.
The two added phase gates are each other's inverses and conjugates; the
same AND--phase--uncompute construction controls each of them exactly.
All numerical entries in these examples meet condition (1).
\end{proof}
The equality for $\mathcal G_\theta$ is over that same gate set on both
sides. It does not assert gate-set independence for arbitrary finite
gate sets, nor exact simulation of arbitrary channel matrices supplied
without unitary dilations.

%% file: circuit-Ct.tex
\begin{figure}[htbp]
  \centering
  \begin{tikzpicture}[x=1cm,y=.82cm]
    \foreach \y in {0,1.25,2.5}{\draw[circuit wire] (0,\y)--(9.8,\y);}
    \node[circuit label,anchor=east] at (-.15,2.5) {$A$};
    \node[circuit label,anchor=east] at (-.15,1.25) {$E:\ket{0}$};
    \node[circuit label,anchor=east] at (-.15,0) {$B$};
    \node[circuit gate,minimum width=1.95cm,minimum height=1.7cm]
      at (1.8,1.875) {$D_t$\\[-1pt]\scriptsize unitary dilation};
    \node[circuit auxiliary,minimum width=2.7cm,minimum height=1.7cm]
      at (5.9,.625) {$U_{\overline{W}_t}^{\dagger}$\\[1pt]
      \scriptsize ideal circuit for\\[-1pt]\scriptsize right Kraus factors};
    \node[circuit tag] at (8.75,1.25) {$E=0$};
    \node[circuit label,anchor=west] at (9.95,2.5) {$A$};
    \node[circuit label,anchor=west] at (9.95,0) {$B$};
  \end{tikzpicture}
  \begin{gather*}
    D_t\bigl(\ket{v}_A\ket{0}_E\bigr) = \sum_a K_{t,a}\ket{v}_A\ket{a}_E,
    \\
    U_{\overline{W}_t}\bigl(\ket{v}_B\ket{0}_E\bigr) = \overline{W}_t\ket{v}, \\
    C_t = \bra{0}_E \bigl(U^{\dagger}_{\overline{W}_t}\bigr)_{BE} \bigl(D_t\bigr)_{AE} \ket{0}_E
        = \sum_a K_{t,a} \otimes \overline{J^{\dagger}_{t,a}}.
  \end{gather*}
  \caption{The ideal doubled step $C_t$. Gates run left to right on a
    possibly entangled input on $A,B$. The blue dilation writes the Kraus
    label into $E$; the purple circuit acts on $B,E$ to supply the right
    Kraus factors $\overline{J_{t,a}^{\dagger}}$.
    The green $E=0$ tag selects the clean block, not a measurement.
    $U_{\overline{W}_t}$ is an exact abstract unitary extension of
    $\overline{W}_t$, with no assumed exact synthesis over the gate set.}
  \label{fig:circuit-Ct}
\end{figure}

%% file: circuit-Chat-t.tex
\begin{figure}[htbp]
  \centering
  \begin{tikzpicture}[x=1cm,y=.82cm]
    \foreach \y in {0,1.15,2.3,3.45}{\draw[circuit wire] (0,\y)--(9.8,\y);}
    \node[circuit label,anchor=east] at (-.15,3.45) {$A$};
    \node[circuit label,anchor=east] at (-.15,2.3) {$E:\ket{0}$};
    \node[circuit label,anchor=east] at (-.15,1.15) {$B$};
    \node[circuit label,anchor=east] at (-.15,0) {$Z:\ket{0}$};
    \node[circuit gate,minimum width=1.95cm,minimum height=1.65cm]
      at (1.8,2.875) {$D_t$\\[-1pt]\scriptsize unitary dilation};
    \node[circuit auxiliary,minimum width=2.7cm,minimum height=2.65cm]
      at (5.9,1.15) {$Q_t^\dagger$\\[2pt]
      \scriptsize inverse of the\\[-1pt]\scriptsize implemented circuit};
    \node[circuit tag] at (8.75,2.3) {$E=0$};
    \node[circuit tag] at (8.75,0) {$Z=0$};
    \node[circuit label,anchor=west] at (9.95,3.45) {$A$};
    \node[circuit label,anchor=west] at (9.95,1.15) {$B$};
  \end{tikzpicture}
  \begin{gather*}
    \widehat{C}_t = \bra{0_E 0_Z} \bigl(Q_t^{\dagger}\bigr)_{BEZ} \bigl(D_t\bigr)_{AE} \ket{0_E 0_Z}
        = \sum_a K_{t,a} \otimes H^{\dagger}_{t,a},
    \\
    H_{t,a} = \bra{a_E 0_Z} Q_t \ket{0_E 0_Z},
    \qquad H_t = \sum_a H_{t,a} \otimes \ket{a}, \\
    \norm{H_t - \overline{W}_t} \le \varepsilon
    \quad\Longrightarrow\quad
    \norm{\widehat{C}_t - C_t} \le \varepsilon,
    \qquad
    \norm{\widehat{C}_t} \le 1.
  \end{gather*}
  \caption{The compiled doubled step $\widehat{C}_t \approx C_t$.
    The purple box replaces the ideal inverse in \cref{fig:circuit-Ct}
    by the inverse of the implemented circuit $Q_t$.
    The left Kraus factors $K_{t,a}$ remain exact; the right Kraus
    factors $H_{t,a}^{\dagger}$ act on $B$.
    Label sums include $*$ and padding labels, with zero left Kraus factors.
    Green tags are clean-block conditions, not measurements: the physical
    output is $\widehat{C}_t\ket{\psi}_{AB}\ket{00}_{EZ}$ plus a component
    with nonzero $EZ$, retained coherently.}
  \label{fig:circuit-Chat-t}
\end{figure}

%% file: appendix-auxiliary.tex
\subsection{Uniform implementation of the right Kraus factors}
\label{app:aux-synthesis}

We prove Lemma~\ref{lem:rql-auxiliary}. A clean block is the compression
of an actual unitary between specified input and output subspaces, with
ancillary registers fixed to zero; the input and output subspaces may
have different dimensions. Compression is only notation: every register
is retained coherently. Throughout this proof, dimensions and running times are polynomial
in $n$, and the required error tolerances are inverse polynomial.

\paragraph{Approximating rotations, including their phases.}
All variable one-qubit rotations can be taken in $\mathrm{SU}(2)$.
There is a fixed, finite, inverse-closed set of determinant-one
Clifford+$T$ circuits whose generated subgroup is dense in $\mathrm{SU}(2)$.
Indeed, the determinant map
on the group generated by $H,T$ has finite image, so its kernel is
finitely generated and has finite index. Projective universality of
$H,T$~\cite{BoykinEtAl99} implies that the projective closure of this
kernel is all of the connected group $\mathrm{PU}(2)$; its closure in
$\mathrm{SU}(2)$ is therefore $\mathrm{SU}(2)$.
Apply the phase-sensitive version of space-efficient Solovay--Kitaev,
\cite[Theorem~4.5]{MelkebeekWatson12}, to this fixed set.
It approximates the specified matrix in operator norm, including its phase, with
$O(\log(1/\xi))$ generation space for error $\xi$.
Scalar phases, when needed, can be realized on the zero state of a
helper qubit using a determinant-one diagonal rotation; this helper is
reused and required to be zero in the selected block.
It is never reset: leakage from approximate rotations is included in
the telescoping error bound for the full circuit relative to the ideal diagonal rotations.
This also supplies phase-correct implementations of the controlled
phase gates in a Fourier transform. All controls and inverses are formed
\emph{after} choosing the finite circuit that implements the approximation, using the
exact gate implementations required by the theorem.
The polynomially many gates used to implement the right Kraus factors
can be approximated with a smaller inverse-polynomial error per gate,
by the unitary telescoping bound.

\paragraph{Reference entries as two-sided computations.}
To implement the recurrence for $\sigma_t$, start in
$I/d$. Apply each source channel, then replace the state by
$I/d$ with probability $1/(T+1)$. For this reference-state experiment, the coin
and all gates may be approximated sufficiently accurately.
Their accumulated trace-distance error is bounded by the sum of the
per-step errors.
For a diagonal entry, or the real or imaginary part of an off-diagonal
entry, measure the corresponding one- or two-coordinate Hermitian
observable. Its norm is at most one, and it is implementable using
reversible basis permutations and constant-dimensional rotations.
Polynomially many independent trials with reused workspace and a
classical sum estimate the entry to any inverse-polynomial additive
accuracy, with constant error probability.

Thus, for an entry component $y\in[-1,1]$, the promise problem
$y\ge z+h$ versus $y\le z-h$ is in $\BQL$ for inverse-polynomial $h$.
The established $\BQL=\BQUL$ theorem~\cite[Theorem~1]{FeffermanRemscrim21},
followed by space-efficient unitary error reduction
\cite[Theorem~1.1 and Corollary~1.2]{FeffermanKLMN16} and
two-sided gate approximation, supplies a uniform Clifford+$T$ comparator
with error at most any inverse-polynomial $\zeta$ on promised inputs.
It is a valid unitary circuit on every input, including thresholds in the
promise gap; no correctness assertion there is needed.
This is a quantum estimation argument, not an assertion that the entries
of $\sigma_t$ are classically computable in logarithmic space.

\paragraph{Representing the reference matrix in a unitary circuit.}
Let $b_z$ be the actual acceptance probability of this comparator.
Compute, apply the sign $+1$ on acceptance and $-1$ on rejection, and
uncompute with its actual inverse. The all-zero work block is exactly
$2b_z-1$. For a power-of-two grid of $N$ midpoint thresholds in $[-1,1]$,
prepare and unprepare a uniform grid register. The resulting scalar block
$a(y)$ obeys
\begin{equation}\label{eq:aux-threshold-bound}
 |a(y)-y|\le 2h+\frac4N+2\zeta.
\end{equation}
Indeed, the integral of $\operatorname{sgn}(y-z)$ against the uniform
measure is $y$; the grid contributes at most $2/N$, the thresholds within
$h$ of $y$ contribute at most $2h+2/N$, and comparator errors outside
that interval contribute at most $2\zeta$.

These scalar circuits can be addressed coherently without quantum access
to a classical circuit generator. Enumerate all polynomially many tuples
$(t,i,j,\text{component},z)$ and emit the corresponding fixed-input
comparison circuit controlled on equality with the address registers.
The address is unchanged. Polynomial enumeration and recomputation use
$O(\log n)$ space.

An equal superposition of the real and imaginary constructions, with an
exact phase $i$ on the latter, gives a scalar entry block $B_{ij}/2$.
Prepare a uniform row register, apply the scalar block controlled on row
and input column, swap these two registers, and unprepare the old column.
Its clean matrix block is $B/(2d)$: on column $j$ the clean output is
$(2d)^{-1}\sum_i B_{ij}\ket i$.
A selector qubit in $\ket+$ chooses between this circuit and its
inverse. Selecting the $\ket+$ output component gives the
exact Hermitian block
\[
 A_t=\frac{S_t}{2d},\qquad
 S_t=\frac{B+B^\dagger}{2},\qquad
 \norm{S_t-\sigma_t}\le\mu.
\]
For example, take $h=\zeta=\mu/(32d)$ and a power-of-two
$N\ge64d/\mu$ in~\eqref{eq:aux-threshold-bound}.
The component errors are at most $3\mu/(16d)$, so the complex entry
errors and the bound $\norm{B-\sigma_t}\le d\max_{ij}|B_{ij}-(\sigma_t)_{ij}|$
give the displayed inequality. Choose $\mu\le\lambda/2$, where
$\lambda=1/(d(T+1))$. Then
$\operatorname{spec}(S_t)\subseteq[\lambda/2,2]$.

\paragraph{Square roots and inverse square roots without storing the matrix.}
Let $U$ be the actual unitary with Hermitian clean block $A_t$, and let
$P$ be its zero-scratch projector. Define
\[
 R=\begin{pmatrix}0&U\\U^\dagger&0\end{pmatrix},\qquad
 P'=\ket+\bra+\otimes P,\qquad Q=(2P'-I)R.
\]
$R$ is a Hermitian unitary and its compression to $P'$ is $A_t$.
For an eigenvalue $a$ of $A_t$, the two-reflection walk $Q$ has phases
$\pm\arccos a$ on the associated invariant plane.
Perform coherent phase estimation of $Q$, rotate a flag qubit so that its $\ket0$ amplitude is
$f(\cos\theta)$, and uncompute the phase estimation.
If $|f|\le1$ and $f$ has Lipschitz constant $L$, an estimate accurate
to $\delta$ except with probability $\nu$ gives clean-block error at most
$L\delta+2\nu$. Both phases in the invariant plane have the same
cosine. With inverse-polynomial $\delta,\nu$, ordinary phase estimation
uses $O(\log n)$ qubits and polynomially many controlled powers of $Q$.
The estimates are retained and uncomputed, never measured.
Rotations depending on an estimate are implemented by enumerating its
polynomially many bit strings and controlling each fixed rotation circuit on
equality. The circuit generator therefore needs only classical addresses.
For each flag rotation the entire $\mathrm{SU}(2)$ matrix, including
$\sqrt{1-f^2}$, is approximated to the prescribed gate accuracy.
Quadratically finer scalar precision suffices near $f=\pm1$ and still
uses only $O(\log n)$ bits.

Put $c(x)=\min\{2,\max\{\lambda/2,2dx\}\}$ and choose
\[
 f_+(x)=\sqrt{c(x)/2},\qquad
 f_-(x)=\sqrt{\lambda/2}\,/\sqrt{c(x)}.
\]
These functions are bounded, have polynomially bounded Lipschitz constants,
and have values computable to inverse-polynomial accuracy in logarithmic space.
Their clean blocks approximate $S_t^{1/2}/\sqrt2$ and
$\sqrt{\lambda/2}\,S_t^{-1/2}$, respectively.
For $\norm{S_t-\sigma_t}\le\mu\le\lambda/2$,
\begin{equation}\label{eq:aux-root-stability}
 \norm{S_t^{1/2}-\sigma_t^{1/2}}\le\frac\mu{\sqrt\lambda},
 \qquad
 \norm{S_t^{-1/2}-\sigma_t^{-1/2}}
       \le\frac{2\mu}{\lambda^{3/2}}.
\end{equation}
For the first bound, the difference of roots solves the Sylvester equation
with coefficients $S_t^{1/2},\sigma_t^{1/2}$; its integral solution has
norm at most $\mu/(\sqrt{\lambda/2}+\sqrt\lambda)$.
The second follows by multiplying this root difference on the left and
right by the inverse roots. Thus both unscaled root errors can meet
any inverse-polynomial error tolerance. No eigenvectors are computed or
stored.

\paragraph{Combining the factors to approximate $W_t$.}
The construction below combines a square root, the inverse of the
unitary dilation $D_t$, and an inverse square root, as in the Petz-isometry implementation of
Gily\'en et al.~\cite[Eq.~(9) and Theorem~1]{GilyenLloydMarvianQuekWilde22}.
Section~\ref{sec:doubling-framework} identifies the matching Kraus
factors and the different completion used here. Our preceding reference
construction supplies logarithmic-space uniform access; the exact
left Kraus factors in the main simulation preserve exact rejection.

Pad the number of source Kraus labels to a power of two $r$.
Prepare a uniform label $L$, copy it into the environment $E$, apply
$D_t^\dagger$ to data and $E$, and require $E=0$ in the output block.
The retained label gives exactly
\[
 \frac{F_t}{\sqrt r},\qquad
 F_t=\sum_a K_{t,a}^\dagger\otimes\ket a,
 \qquad \norm{F_t}\le\sqrt r.
\]
Compose this block with the root blocks before and after it, and add a
separate branch for $J_{t,*}\otimes\ket *$.
For a sum, prepare a selector with amplitudes proportional to the square
roots of the coefficient magnitudes, select the component unitary, and
unprepare the selector. Compression yields the sum divided by the sum
of those magnitudes. There are only constantly many branches here.
Products of the constant number of constituent blocks use distinct clean
scratch, or the failure-counter construction in the main proof.
The resulting block is
\[
 s_0(W_t+\Delta),\qquad s_0\ge n^{-O(1)}.
\]
All normalization factors are known scalar functions of $d,T,r$.
For instance the ideal main-branch coefficient is
$2\sqrt{r/(\kappa\lambda)}$ and the star-branch coefficient is
$\sqrt2$, using $\eta/(\kappa d)=\lambda$.
The common scale is their sum's reciprocal.
Errors in roots and scalar rotations change the unscaled block by at
most a polynomial factor times their errors: this follows by a
constant-length product telescoping bound and
$\norm{F_t}\le\sqrt r$, $\norm{\sigma_t^{-1/2}}\le\lambda^{-1/2}$.
Consequently any specified inverse-polynomial bound on $\norm\Delta$
is attainable. Errors in preparing the selector are also bounded after division by
the known normalization factor $s_0$.

\paragraph{Removing the normalization factor.}
Choose a polynomially bounded integer $m\ge1$ for which
$s=\sin(\pi/(4m+2))\le\min\{s_0,1/2\}$ and reduce the selected block's amplitude
by a factor $s/s_0$. First analyze an ideal rotation giving this factor.
Rectangular oblivious amplitude amplification, using the input and
output clean-subspace reflections, transforms the block
$s(W_t+\Delta)$ as follows. Write $W_t+\Delta=VH$ in polar form.
If $\norm\Delta\le\gamma\le1/4$, then
\[
 \norm{V-W_t}\le2\gamma,\qquad
 \operatorname{spec}(H)\subseteq[1-\gamma,1+\gamma].
\]
After $m$ iterations the block is, with its known sign corrected,
\[
 V f(H),\qquad
 f(x)=\sin\bigl((2m+1)\arcsin(sx)\bigr).
\]
Here $f(1)=1$ and $|f'(x)|<3$ on that interval, since
$(2m+1)s\le\pi/2$ and $sx\le5/8$.
Thus $\norm{Vf(H)-W_t}\le5\gamma$.
Use $\gamma\le\varepsilon/20$. Approximate this extra rotation so
that the complete amplified circuit changes by at most $\varepsilon/4$;
per-call error $O(\varepsilon/(m+1))$ suffices by unitary telescoping.
The preceding approximation errors can be included in $\Delta$ before
amplification. The resulting block is within $\varepsilon$ of $W_t$,
with all Kraus labels retained and no additional normalization factor.
Thus this implementation avoids a fixed loss at every source step.

Finally conjugate this actual circuit gate by gate. For $\Gamma$, use its exact
complex-conjugation circuits; more generally, use condition (3) of
Theorem~\ref{thm:rql-same-gates},
to obtain $Q_t$ and its labelled block within $\varepsilon$ of
$\overline W_t$. The initial states and projectors selecting the block have real entries, so
conjugation commutes with taking this block.

\paragraph{An explicit precision budget.}
The following choices make the error estimates quantitative. Let
$0<\varepsilon\leq1$ be the desired error in the labelled block of
Lemma~\ref{lem:rql-auxiliary}, and set
\[
 \lambda=\frac1{d(T+1)},\qquad
 K=2\sqrt{\frac r{\kappa\lambda}}+\sqrt2=\frac1{s_0},
 \qquad \gamma=\frac{\varepsilon}{20}.
\]
Here $r$ is the padded number of source Kraus labels. All tolerances
in Table~\ref{tab:aux-precision} may be rounded down to dyadic numbers
within a factor of two of the stated bound. The grid size is rounded up
to a power of two. These choices require only logarithmically many bits.

\begin{table}[tbp]
\centering
\small
\setlength{\tabcolsep}{5pt}
\renewcommand{\arraystretch}{1.22}
\begin{tabularx}{\textwidth}{@{}p{0.25\textwidth}p{0.33\textwidth}X@{}}
\toprule
Quantity & Sufficient choice & Resulting error or cost \\
\midrule
Reference-matrix error
 & $\mu\leq\gamma\lambda/(16K)$
 & $\norm{S_t-\sigma_t}\leq\mu$ \\
Threshold comparison
 & $h=\zeta=\mu/(32d)$; $N\geq64d/\mu$
 & Gives the preceding matrix bound by~\eqref{eq:aux-threshold-bound} \\
Reference experiment
 & Trace-norm error at most $h/4$ per trial; $\lceil128/h^2\rceil$ trials
 & Constant-error comparison with threshold gap $h$ \\
Root phase estimation
 & $\delta\leq\gamma\lambda/(32dK)$; $\nu\leq\gamma/(32K)$
 & Functional-calculus error at most $\gamma/(8K)$ per normalized root \\
Polar amplification
 & Integer $K\leq m\leq2K+2$; $s=\sin(\pi/(4m+2))$
 & Error at most $5\gamma=\varepsilon/4$ before final gate synthesis \\
Variable rotations
 & Operator-norm error $\xi\leq\varepsilon/(4M)$ per occurrence
 & At most $\varepsilon/4$ additional error in the full amplified circuit \\
\bottomrule
\end{tabularx}
\caption{One sufficient error budget for a single right Kraus isometry.
$M\geq1$ bounds the number of variable-rotation occurrences in the
complete amplified circuit before those rotations are synthesized.
Repeated occurrences, including those in inverses and controlled
circuits, are counted separately.}
\label{tab:aux-precision}
\end{table}

For the reference experiment, a trace-norm error of $h/4$ changes the
expectation of a norm-one observable by at most $h/4$. Hoeffding's
inequality gives empirical error at most $h/4$, except with probability
at most $2e^{-4}<1/3$, using the stated number of trials. Thus a promised
distance of at least $h$ from the threshold leaves a constant-probability
correct comparison. The required trace-norm bound follows by telescoping
the implemented reference channels: for at most $T$ channels, an error
at most $h/(4T)$ in diamond norm per channel suffices. For a reference trial with $g$ approximated unitary gate occurrences,
per-gate operator-norm error at most $h/(8g)$ suffices; the intervening
channels contract trace distance.
Applying \cite[Theorem~1]{FeffermanRemscrim21} gives a unitary comparator
with constant gap. Then \cite[Corollary~1.2]{FeffermanKLMN16}, with
$p=\lceil\log_2(2/\zeta)\rceil$, reduces its error to at most $\zeta/2$.
A further two-sided synthesis error of at most $\zeta/2$ in its acceptance
probability gives the actual comparator error $\zeta$ used above.

To check the root budget, both $f_+$ and $f_-$ have Lipschitz constant
at most $2d/\lambda$: on the unclipped interval their derivatives are
bounded respectively by $d/\sqrt\lambda$ and $2d/\lambda$, and clipping
does not increase these bounds. Hence $L\delta+2\nu\leq\gamma/(8K)$.
Equation~\eqref{eq:aux-root-stability} changes either normalized root by
at most $2\mu/\lambda\leq\gamma/(8K)$. Each implemented root block
therefore differs from its block defined using $\sigma_t$ by at most
\[
 \varepsilon_{\rm root}=\frac{\gamma}{4K}.
\]
Write $a=2\sqrt{r/(\kappa\lambda)}$ and $b=\sqrt2$, so $a+b=K$.
The main branch is a product of two normalized root blocks and the
contraction $F_t/\sqrt r$; the additional branch uses one normalized
inverse root. Every factor is a contraction. Product telescoping
therefore bounds the error in the unscaled sum by
\[
 \norm\Delta\leq (2a+b)\varepsilon_{\rm root}\leq2K\varepsilon_{\rm root}=\gamma/2.
\]
The polar-amplification bound proved above consequently applies.
The chosen $m=O(K)$ guarantees $s\leq s_0$ and $s\leq1/2$.

In this calculation the comparator circuits are already finite
Clifford+$T$ circuits, while the remaining variable rotations are first
analyzed as exact matrices. The number $M$ of their occurrences is
fixed by the preceding circuit construction and does not depend on
their eventual synthesis accuracy. Approximate all of them, including
the selector, Fourier-transform, root-flag and attenuation rotations,
with the final tolerance $\xi$. The unitary telescoping bound adds at
most $M\xi\leq\varepsilon/4$ to the amplified circuit's error.
Thus the total labelled-block error is at most $\varepsilon/2$, leaving
slack within the required $\varepsilon$.
In particular, the main simulation's choice
$\varepsilon=1/(12Td)$ makes the accumulated doubled-step error at most
$T\varepsilon=1/(12d)$.

The table specifies tolerances rather than claiming a particular
polynomial degree for the imported two-sided simulator. The reference
trials, address enumeration, phase-estimation powers and $O(K)$
amplification calls each have length polynomial in $n,d,T,r$ and
$1/\varepsilon$. For example, ordinary phase estimation can attain
$(\delta,\nu)$ with $O(1/(\delta\nu))$ controlled walk calls and
$O(\log(1/(\delta\nu)))$ estimation qubits. The imported comparator has
polynomial length and logarithmic space by the cited theorems.
Consequently $M$ and the fully synthesized circuit length are
polynomial, while every address, counter, precision register and
rotation generator uses $O(\log n+\log(1/\varepsilon))$ space.
There are only constantly many nested constructions. For the
inverse-polynomial tolerances used here, the total space is $O(\log n)$.

\paragraph{Time and space bounds.}
All lower bounds on eigenvalues and all error tolerances above are
inverse polynomial.
The number of address tuples, phase-estimation calls, amplification
iterations, and gates is polynomial. Scalar calculations keep only a
constant number of $O(\log n)$-bit registers; loop bodies and gate sequences
are regenerated when needed.
Only a constant number of constructions are nested: the two-sided
comparison circuit, reference block, root block, construction of $W_t$, and
amplification. Each level uses logarithmic workspace, hence their sum
does also. Repeated powers are powers of full unitaries on the same
registers and do not allocate new workspace on each call.
The reference estimator uses the original measured computation; the
two-sided measurement-elimination theorem is used as an established result.
There is no recursive invocation of the one-sided simulator and no
assumption of a classically accessible reference-state matrix.
This proves Lemma~\ref{lem:rql-auxiliary}.

%% file: result-bql.tex
\clearpage
\section{The conditional \texorpdfstring{$\BQL$}{BQL} application}

\label{sec:doubling-bql}

Here $\BQL$ and $\BQUL$ refer to the measured and unitary
polynomial-time, logarithmic-space models over
$\Gamma=\{H,T,\mathrm{CNOT}\}$, with acceptance at least $2/3$ on
yes-instances and at most $1/3$ on no-instances. We separate the
two-sided correctness argument from the implementation of its clean
blocks.
The implementation of the right Kraus factors in
Appendix~\ref{app:aux-synthesis} uses the
established equality $\BQL=\BQUL$; the result below is a conditional
doubling simulation, rather than an independent proof of that equality.

\begin{proposition}[Conditional bounded-error simulation]
\label{prop:doubling-bql}
Suppose that, for every polynomial-time, logarithmic-space source
computation over $\Gamma$ and every
inverse-polynomial $\varepsilon>0$, there are polynomial-size,
logspace-uniform unitary circuits $V_t$ over $\Gamma$, on $O(\log n)$
qubits, whose zero-scratch blocks $\widehat C_t$ satisfy
\[
 \norm{\widehat C_t-C_t}\le\varepsilon.
\]
Then the reweighted doubling construction, preceded by measured source
error reduction and followed by unitary amplification, simulates each
$\BQL$ computation in $\BQUL$.
\end{proposition}

\begin{proof}
Fix a source machine and an input of length $n$.

\paragraph{Reduce source error before doubling.}
Run $2n+1$ independent repetitions of the original machine, resetting
and reusing its workspace, and take their majority with a classical
counter. Hoeffding's bound gives error at most $e^{-(2n+1)/18}$.
The counter uses $O(\log n)$ bits, and the running time remains
polynomial. Let $d$ be the padded power-of-two configuration dimension
of this amplified source, and $T\ge1$ a polynomial bound on its running
time. Both are polynomial in $n$. Therefore, for all sufficiently large
$n$, its acceptance probability satisfies
\begin{equation}\label{eq:bql-source-error}
 p\ge2/3\quad\text{on yes-instances},\qquad
 p\le\delta:=\frac1{10^6d^2}\quad\text{on no-instances}.
\end{equation}
The finitely many smaller input lengths can be handled directly.
The repetition count is fixed before defining $d$, so this argument does
not choose an accuracy circularly through the workspace it requires.

\paragraph{Compose the approximate contractions.}
Use $\varepsilon=1/(1000Td)$. Since every $\widehat C_t$ is a block
of an actual unitary, both $C_t$ and $\widehat C_t$ are contractions.
Consequently
\[
 \norm{\widehat C_T\cdots\widehat C_1-C_T\cdots C_1}
 \le T\varepsilon=:\Delta=\frac1{1000d}.
\]
Give every step the same number of scratch qubits by adding unused qubits. To realize the product
coherently, reuse this scratch and add a
counter modulo a power of two larger than $T$. Apply every $V_t$
unconditionally, then increment the counter if the scratch is nonzero.
Expand the product using the projectors onto zero and nonzero scratch.
Final counter zero selects exactly those terms with zero scratch after
every step: the counter can increase at most $T$ times and cannot wrap
to zero.
This is the
construction proved in~\eqref{eq:adjoint-counter-step}.

The good subspace requires counter zero, scratch zero, and acceptance
on the first data register. If its probability is $\widehat q$, then
\[
 \left|\sqrt{\widehat q}-\sqrt{q_*}\right|\le\Delta.
\]
By~\eqref{eq:doubling-signal-bounds} and~\eqref{eq:bql-source-error},
\begin{align*}
 \sqrt{\widehat q}&\ge\frac1{3\sqrt d}-\frac1{1000d}
                 \ge\frac1{4\sqrt d}
 &&\text{on yes-instances},\\
 \widehat q&\le(\sqrt\delta+\Delta)^2
             =\frac4{10^6d^2}
 &&\text{on no-instances}.
\end{align*}
The first line uses $p\ge2/3$ and $e<4$.

\paragraph{Amplify with both errors controlled.}
Append a Hadamard coin and require coin one in the good predicate.
The resulting preparation $A$ has success $q=\widehat q/2\le1/2$.
Thus $q\ge1/(32d)$ on yes-instances and
$q\le2/(10^6d^2)$ on no-instances. All zero-state and good-subspace
reflections, controls, and actual inverses are exactly implementable
over $\Gamma$.

Prepare and retain a uniform choice $j\in\{0,\ldots,M-1\}$, where
$M=4d$, and apply $j$ Grover iterations to $A\ket0$
\cite{BrassardHoyerMoscaTapp02}.
For $q=\sin^2\theta>0$, the averaged success is
\[
 \frac1M\sum_{j=0}^{M-1}\sin^2((2j+1)\theta)
 =\frac12-\frac{\sin(4M\theta)}{4M\sin(2\theta)}.
\]
On yes-instances, $q\in[1/(32d),1/2]$ gives
$\sin(2\theta)\ge1/(4\sqrt d)$, so success is at least $1/4$.
On no-instances, the bound
$|\sin(k\theta)|\le k|\sin\theta|$ for integer $k\ge1$ gives
\[
 \sin^2((2j+1)\theta)\le(2j+1)^2q
 <(8d)^2\frac2{10^6d^2}=\frac{128}{10^6}.
\]
This also holds when $q=0$. Four independent coherent copies followed
by a reversible OR of their good predicates therefore accept with
probability at least $1-(3/4)^4=175/256>2/3$ on yes-instances, and
less than $512/10^6<1/3$ on no-instances.
Only the final output qubit is measured.

The reused scratch, counter, choice register, and a constant number of
copies require $O(\log n)$ space. All repetition counts and circuit
lengths are polynomial, and the loops and inverse circuits can be
generated in logarithmic space.
\end{proof}

\paragraph{Implementation and logical dependency.}
Lemma~\ref{lem:rql-auxiliary} and the clean-block construction
in~\eqref{eq:rql-actual-step} supply the premise of
Proposition~\ref{prop:doubling-bql} for $\Gamma$.
Their reference-entry comparators use the established two-sided
measurement-elimination theorem of Fefferman and Remscrim
\cite{FeffermanRemscrim21}, as explained in
Appendix~\ref{app:aux-synthesis}. Thus the present implementation uses
$\BQL=\BQUL$ as an input. An independent doubling proof of that
equality would require computing the reference entries and implementing
the right Kraus factors without using that theorem. This dependence does not affect the
one-sided theorems in the main text, which use the established two-sided
result as a premise.

%% file: appendix-background.tex
\clearpage
\section{Further quantum-space background}
\label{app:space-background}
This appendix expands the quantum-space background of
Section~\ref{sec:intro-background}. Figure~\ref{fig:logspace-map}
in that section summarizes the established containments.

\paragraph{The quantum space landscape.}
Watrous established the foundations of space-bounded quantum complexity and its classical simulation~\cite{watrous1999space,Watrous03}.
For standard finite algebraic gate sets, these results place bounded-error quantum logspace $\BQL$ inside classical polynomial time and deterministic $O(\log^2 n)$ space; the corresponding unbounded-error quantum and probabilistic space classes coincide.
These upper bounds leave room for quantum advantages over classical logarithmic space.
One early example is Watrous' unitary simulation of random walks, which placed undirected connectivity in $\RQUL\cap\mathsf{co}\RQUL$~\cite{Watrous01}.
A different characterization was obtained by Jozsa, Kraus, Miyake and Watrous: polynomial-size, logspace-uniform nearest-neighbour matchgate circuits have the same computational power as polynomial-time unitary quantum logspace~\cite{JozsaKrausMiyakeWatrous10}.

Linear algebra provides another useful description of this power.
Ta-Shma gave quantum logspace algorithms for approximating inverses of polynomially well-conditioned matrices~\cite{TaShma13}.
Fefferman and Lin obtained unitary algorithms and matching completeness results for suitably normalized matrix-inversion and minimum-eigenvalue problems~\cite{FeffermanLin18}.
The normalization and condition-number promises are essential: these results do not put unrestricted matrix inversion in quantum logspace.
Space-efficient amplification is a further ingredient.
Fefferman, Kobayashi, Lin, Morimae and Nishimura showed how to reduce error to $2^{-p}$ using only $O(\log(p/\Delta))$ additional qubits, where $\Delta$ is the original completeness--soundness gap~\cite{FeffermanKLMN16}.
As a more recent application of quantum linear algebra, Apers and Edenhofer gave a quantum logspace procedure for counting directed $s$--$t$ paths under polynomial bounds on the number of paths starting at $s$ and ending at $t$~\cite{ApersEdenhofer25}.
The promise concerns these larger families of paths, not merely the number of $s$--$t$ paths.

\paragraph{Removing measurements with two-sided error.}
Fefferman and Remscrim proved $\BQL=\BQUL$, and the analogous equality for space-constructible bounds $S\geq\log n$~\cite[Theorem~1]{FeffermanRemscrim21}.
Their result covers general computations, including classical control conditioned on measurement outcomes.
The proof relates their acceptance probabilities to well-conditioned linear-algebraic problems that admit unitary space-efficient algorithms.
Independently, Girish, Raz and Zhan gave a unitary quantum logspace algorithm for powering contraction matrices~\cite{GirishRazZhan21}.
They used it to simulate \emph{unital} quantum computations, including circuits with unitary gates and intermediate measurements but without a separate classical memory that stores and acts on the outcomes~\cite[Section~5]{GirishRazZhan21}.
These are distinct model assumptions: a reset channel, for example, is not unital.

For $S=O(\log n)$ and $T=\poly(n)$, these simulations have polynomial running time; their dependence on a general space bound is exponential in $S$.
Girish and Raz obtained a different trade-off: for circuits with unitary operations and intermediate measurements, time $T$ and space $S\geq\log T$ can be replaced by time $T\,\poly(S)$ and space $O(S\log T)$~\cite{GirishRaz22}.
Zhandry proved lower bounds for black-box measurement-elimination procedures, showing limitations on simultaneously preserving time and space in that setting~\cite{Zhandry24}.
Postselection gives a related comparison: Le Gall, Nishimura and Yakary\i lmaz proved $\mathsf{PostBQL}=\PL$, and Tani showed that intermediate measurements and intermediate postselections can be deferred in the space-bounded, bounded-error postselection model~\cite{LeGallNishimuraYakaryilmaz21,Tani22}.
These results concern their respective computational models; postselection is not an operation available to an ordinary $\RQUL$ machine.

\paragraph{Why one-sided error is different.}
An additive approximation to an acceptance probability preserves a constant two-sided gap, but need not preserve perfect soundness.
Fefferman and Remscrim explicitly isolated the remaining one-sided bounded-error question as Open Problem~1~\cite[Section~6]{FeffermanRemscrim21}.
Their partial characterization includes
\[
  \RQUL\subseteq\RQL\subseteq\mathsf{coQ_UMAL}_1,
\]
where the last class is the complement of unitary quantum Merlin--Arthur logspace with perfect completeness~\cite[Theorem~32 and Lemma~36]{FeffermanRemscrim21}.
They also established equality between general and unitary quantum logspace with one-sided \emph{unbounded} error~\cite[Theorem~41]{FeffermanRemscrim21}; allowing an arbitrarily small nonzero accepting probability is a different requirement from retaining an inverse-polynomial signal.

Le Gall, Liu and Wang subsequently found natural complete promise problems for $\mathsf{co}\RQUL$: given descriptions of logarithmic-space state-preparation circuits, distinguish equality of the output states from inverse-polynomial separation in trace distance or squared Hilbert--Schmidt distance~\cite[Theorem~1.1]{LeGallLiuWang23}.
Their algorithms use space-efficient quantum singular value transformation.
Their gate-set convention is relevant here: they allow single-qubit gates with space-computable amplitudes and explicitly discuss the difficulty of changing gate sets under one-sided error~\cite[Remark~2.7]{LeGallLiuWang23}.
For two-sided bounded error, the space-efficient Solovay--Kitaev theorem makes the class independent of the gate set under the usual universality
and computability assumptions, with each gate's inverse available~\cite{MelkebeekWatson12}.
It does not establish gate-set independence when zero acceptance must
be preserved exactly.
A same-gate-set statement must therefore specify which operations are synthesized exactly and which are approximated.

%% file: appendix-formal.tex
\clearpage
\section{Machine-checked mathematical lemmas}
\label{app:formal-verification}
The Lean project verifies the supporting lemmas in
Table~\ref{tab:formal-coverage}, with the displayed assumptions. It does
not formalize the complexity-class equalities.

\begin{table}[htbp]
\centering
\small
\setlength{\tabcolsep}{4pt}
\renewcommand{\arraystretch}{1.18}
\begin{tabular}{p{.22\linewidth}p{.34\linewidth}p{.37\linewidth}}
\hline
Paper component & Checked statements & Supplied assumptions or remaining steps \\
\hline
Reweighted doubling &
Exact zero preservation for arbitrary right factors; vectorization;
reweighted-step identity and telescoping. &
Finite history expansion and supplied matrices satisfying the inverse
identities. Matrix-root construction is not formalized. \\
Approximation and amplification &
Product error at most $T\varepsilon$; norm and probability perturbation;
averaged-Grover identities and scalar gap margins. &
Per-step errors and norm bounds; the reduction of the circuit to a
Grover rotation and its angle estimate. \\
Failure counter and recovery map &
Zero count selects exactly the all-zero-scratch histories;
normalization and reference recovery of the Kraus family. &
The circuit's history expansion; explicit root, channel-image and
trace-preservation equations. \\
Witness histories &
Kernel equivalence and the singular-value bound in
Lemma~\ref{lem:qmal-history-gap}, with Euclidean direct-sum norms. &
Bounds on interval products, the honest history and the boundary map;
the original rejection lower bound. \\
Exact matrix test &
Paired coefficients, inactive-branch swaps, common phase, padding and
uniform-address matrix identities from
Lemma~\ref{lem:qmal-exact-block}. &
The gate implementation and address-controlled circuit construction. \\
Coherent witness amplification &
Rejection geometry; the $PQP$ norm bound; exact completeness;
constant soundness and finite failure-register norm preservation. &
The exact unitary block, an isometric embedding of initialized inputs,
and the history matrix's lower norm bound. \\
\hline
\end{tabular}
\caption{Scope of the Lean verification. The supplement's
\texttt{formal/README.md} gives the module and theorem names for each row.}
\label{tab:formal-coverage}
\end{table}

\paragraph{Boundary of the verification.}
Quantum-machine semantics, the circuit compiler, the acceptance-probability
bounds~\eqref{eq:doubling-signal-bounds}, uniform reference access,
square-root synthesis, the CM-field encoding, and the complete uniformity
and resource analysis remain ordinary mathematical proofs. For the
witness result, positivity and the analytic bounds for backward channel
maps are also supplied as hypotheses. Passing the Lean check does not
certify these interfaces or the resulting complexity-class equalities.

\paragraph{Reproducing the check.}
The supplement \texttt{formal-verification-source.zip} contains the
\texttt{formal/} project; \texttt{VERIFY-MANIFEST.json} records its
source revision and file hashes. It pins Lean~4.30.0 and all library
dependencies. After installing \texttt{elan}, run from \texttt{formal/}:
\begin{verbatim}
lake exe cache get
python3 verify.py
\end{verbatim}
The checker rebuilds proofs, rejects placeholders and audits foundational
axioms; its report records source hashes. The README gives setup details
and verification boundaries. Preserve the supplied dependency manifest
and rerun the check after edits.